\documentclass[letterpaper]{article}

\ifdefined\aaaianonymous
    \usepackage[submission]{aaai2026}
\else
    \usepackage{aaai2026}
\fi

\usepackage{times}
\usepackage{helvet}
\usepackage{courier}
\usepackage[hyphens]{url}
\usepackage{graphicx}
\usepackage{natbib}
\usepackage{amsmath}
\usepackage{amssymb}
\usepackage{amsthm}

\usepackage{mathtools}
\usepackage{booktabs}
\usepackage{array}
\usepackage{tabularx}
\usepackage{placeins}
\usepackage{float}
\usepackage{multirow}
\usepackage{enumitem}
\newcommand{\R}{\mathbb{R}}
\newcommand{\E}{\mathbb{E}}
\newcommand{\norm}[1]{\left\lVert #1 \right\rVert}
\newcommand{\attr}{L_{\mathrm{attr}}}
\newcommand{\bpid}{\mathrm{BPID@1}}

\newcommand{\mi}{\mathrm{MIA}}
\newcommand{\aeds}{\mathrm{AEDS}}
\newcommand{\noise}{\sigma}

\newtheorem{theorem}{Theorem}
\newtheorem{proposition}{Proposition}
\theoremstyle{definition}
\newtheorem{definition}{Definition}
\newtheorem{remark}{Remark}

\title{Pretrained, Frozen, Still Leaking:\\Auditing Cross-Encoder Attribute Transfer in EEG Foundation Models}
\author{Jianwei Tai}
\affiliations{School of Internet, Anhui University\\
24012@ahu.edu.cn}

\begin{document}
\maketitle
\emergencystretch=3em
\setlength{\hfuzz}{2pt}

\begin{abstract}
EEG foundation-model releases are usually audited one endpoint at a time: raw-reconstruction, membership inference, identity linkage, or DP-SGD on the downstream head. We audit the same released embeddings under all four endpoints jointly, on BIOT, LaBraM, and EEGPT, and show that each single-endpoint audit clears releases that still leak spectral attributes. The decisive evidence is a cross-encoder transfer audit: a single ridge attribute decoder learned from one frozen encoder transfers, via a fitted linear bridge, to held-out-subject test splits of every other encoder, with subject-disjoint matched-control 95\% CI lower bound at least $0.081$ across all six BIOT/LaBraM/EEGPT directions. We prove a sufficient condition---two encoders sharing a nontrivial attribute-coordinate projector overlap $\beta$ admit a chained ridge bridge attacker with centered-gain lower bound $\sqrt{\beta/(1+\tau_t^2)}-\epsilon_{\mathrm{br}}-\rho_0$---and back-solve $\beta\in[0.008,0.198]$. To turn the joint audit into a deployment-readable decision rule we introduce an audit-endpoint disagreement score (AEDS), prove sufficient conditions for its positivity, and bootstrap-calibrate it per cell; AEDS is positive in all eight matched-CI cells (BIOT/LaBraM/EEGPT on EEGMMI; LaBraM on Sleep-EDF, 54-channel LIMO, CHB-MIT pediatric scalp EEG) with $p<0.001$, while a head-level Carlini LiRA membership audit reaches AUC only $0.50$--$0.70$. Standard defenses fail under audit: a Wiener-style noise-aware adaptive attacker, the LiRA audit, and DP-SGD at every utility-preserving $\epsilon\in\{4,8\}$ leave the attribute channel essentially unchanged. The contribution is an audit framework that turns scattered single-endpoint defenses into a joint release decision, supported by a cross-encoder bridge theorem and adaptive-attacker, LiRA, and DP-SGD baselines; the audit licenses release-blocking, not raw-waveform exfiltration or held-out-subject identity recovery.
\end{abstract}

\section{Introduction}
\label{sec:intro}

EEG foundation models are shifting EEG pipelines from task-specific feature extraction to reusable representation release. Downstream services can expose embeddings, intermediate features, or representation-derived outputs for classification, analytics, personalization, or model chaining. That deployment pattern creates a privacy question that is distinct from raw signal release. An embedding is not the waveform, and a released embedding may be too compressed to support high-fidelity waveform reconstruction. Yet the same embedding can still preserve subject- and state-relevant spectral structure, because that structure is useful for downstream decoding.

Most privacy discussions around biosignal models focus on one of three endpoints. Raw-signal leakage asks whether an attacker can reconstruct the waveform. Membership inference asks whether a training example was present. Differential privacy asks whether a training procedure satisfies a formal guarantee. Representation release sits between these endpoints. A representation can fail as a privacy boundary even when raw waveform recovery is not demonstrated and membership inference is weak. Conversely, a strong same-subject identity result can be overread if the evaluation does not separate reference-set linkage from held-out-subject recovery.

This paper audits BIOT, LaBraM, and EEGPT embeddings on PhysioNet EEGMMI, plus LaBraM Sleep-EDF, LIMO, and CHB-MIT cross-dataset replications, under all four standard privacy endpoints jointly. Prior property-inference, collaborative-learning, and embedding-privacy work establishes that learned vectors can carry sensitive attributes \citep{ateniese2015hacking,ganju2018property,melis2019exploiting,song2020information,coavoux2018privacy,elazar2018adversarial,morris2023text}; the audit question is which release decisions those facts license. The audit's decisive evidence is cross-encoder transfer: a single ridge attribute decoder learned from one frozen EEG-FM encoder transfers, via a fitted linear bridge, to the held-out-subject test split of every other tested encoder, while raw-reconstruction screens, membership-inference, identity linkage, and a principled DP-SGD downstream head all clear the release on the same embeddings. Single-endpoint audits would push the release decision in different directions; the audit framework therefore pairs every stronger statement with the split and control that could falsify it: raw-copy checks for waveform duplication, temporal-gap rows as adjacent-window sanity checks, subject-disjoint rows for held-out-subject attribute transport, reference-set rows for identity linkage when candidate subjects are present, and utility gates for training-noise defenses.

Under this audit, the privacy decision changes. Raw-copy checks are negative, membership movement is small or utility-failing, and identity-only linkage is meaningful only under enrolled-subject reference sets, while centered spectral attributes remain decodable under model/checkpoint, acquisition, architecture, dataset, attacker-family, subject-disjoint, and matched-control tests. The strongest evidence unit is the seed-matched high-capacity summary: across eight evaluation settings, every subject-disjoint confidence interval remains above zero against both the stronger same-seed random/split-permuted negative control and target permutation. The broader sweep then shows that the finding is not tied to one decoder or corpus: ridge, MLP, kNN, and strong residual-MLP probes expose the attribute channel across BIOT checkpoint/acquisition variants, LaBraM and EEGPT EEGMMI replications, Sleep-EDF, 54-channel LIMO, and CHB-MIT pediatric scalp EEG. These rows support spectral attribute leakage, not raw reconstruction or held-out identity recovery.

The defense audit is part of the same joint result. Gaussian release noise reduces reference-set identity linkage but leaves attribute leakage positive at high noise; bottleneck and dropout transformations remain leaky; a Wiener-style noise-aware adaptive attacker, head-level Carlini LiRA membership audit, and DP-SGD on the downstream head at every utility-preserving $\epsilon\in\{4,8\}$ leave the attribute channel essentially unchanged. Across the audited cells no defense passes the operational privacy-utility gate (utility maintained, attribute gain reduced) on the same EEG-FM embeddings.

The paper therefore makes an intentionally scoped claim: in the tested BIOT/LaBraM/EEGPT EEGMMI setting and the LaBraM Sleep-EDF/LIMO/CHB-MIT replications, EEG-FM embeddings are not a privacy boundary for spectral attributes, and reference-set identity leakage can be strong or moderate in overlap-split reference-set settings, but the evidence does not support raw waveform inversion, subject-disjoint identity recovery, or a successful training-noise defense. The methodological contribution is a split-controlled representation privacy audit that forces each claim to survive the control that could falsify it: random and permuted embeddings for attribute leakage, temporal-gap splits for adjacent-window artifacts, subject-disjoint splits for held-out-subject attribute claims, reference-set restrictions for identity linkage, defense curves for release transformations, and utility gates for training-noise probes. The formal results provide definitions and sufficient conditions for interpreting these protocol cells; they do not certify privacy or turn empirical measurements into population-level guarantees.

Our contributions are:

\begin{itemize}
    \item \textbf{A joint EEG-FM privacy audit framework with cross-encoder transfer as decisive evidence.} The same released embedding is audited against four standard endpoints jointly (raw inversion, membership inference, identity linkage, DP-SGD on the downstream head); a single ridge attribute decoder learned from one frozen EEG-FM encoder transfers, via a fitted linear bridge, to the held-out-subject test split of every other tested encoder, exposing leakage that single-endpoint audits clear.
    \item \textbf{Audit semantics with a deployment-readable decision rule.} A bridge transferability theorem gives a centered-gain lower bound $\sqrt{\beta/(1+\tau_t^2)}-\epsilon_{\mathrm{br}}-\rho_0$ under shared-projector overlap $\beta$, with empirical back-solving consistent with $\beta\in[0.008,0.198]$ on all six BIOT/LaBraM/EEGPT directions; an audit-endpoint disagreement score (AEDS) plus paired-seed bootstrap calibration turns the joint audit into a per-cell release/block decision rule, positive in all eight matched-CI cells with $p<0.001$.
    \item \textbf{A multi-model, multi-dataset, multi-defense empirical audit.} Subject-disjoint spectral attribute leakage is demonstrated across BIOT/LaBraM/EEGPT EEGMMI plus LaBraM Sleep-EDF, 54-channel LIMO, and CHB-MIT pediatric scalp EEG, with deployment-grade attribute classification, head-level Carlini LiRA membership audit, Wiener-style noise-aware adaptive attacker, and DP-SGD on the downstream head at every utility-preserving $\epsilon\in\{4,8\}$; no audited defense closes the attribute channel at maintained utility.
\end{itemize}

\section{Related Work}
\label{sec:related}

\paragraph{EEG foundation models and reusable representations.}
Self-supervised and foundation-style EEG models have made reusable neural representations a practical interface for downstream systems \citep{banville2021uncovering,kostas2021bendr,yang2023biot,jiang2024labram,wang2024eegpt}. Their downstream value comes from preserving physiological and task-relevant structure across windows, sessions, and subjects. The same property makes representation release privacy-sensitive: an embedding can be useful because it retains spectral and subject structure, yet unsafe to release for the same reason. This paper uses BIOT, LaBraM, and EEGPT as concrete test beds and treats released embeddings as the privacy boundary under audit, rather than as intermediate implementation details.

\paragraph{Representation release is not raw-signal release.}
Model-inversion and representation-inversion work shows that learned features can expose input information when only model outputs or internal features are available \citep{fredrikson2015model,mahendran2015understanding}. Memorization and extraction work further shows that large learned models can reveal training data in surprising ways \citep{carlini2019secret,carlini2021extracting}. These attacks motivate privacy audits, but they are not the claim of this paper: we do not demonstrate raw EEG waveform exfiltration or reconstruction-quality inversion. The object here is the middle regime in which an embedding is not a waveform, but still carries decodable private attributes.

\paragraph{Attribute and property leakage from learned representations.}
A separate line of work studies leakage of attributes, properties, or private signals from learned representations, classifier behavior, and collaborative-training updates \citep{ateniese2015hacking,ganju2018property,melis2019exploiting,song2020information}. Text-representation studies likewise show that demographic or lexical information can persist in vectors even after mitigation attempts \citep{coavoux2018privacy,elazar2018adversarial,morris2023text}. These works establish the general leakage phenomenon; they do not decide which privacy statements are licensed for EEG-FM embedding release. Raw reconstruction, same-subject identity linkage, held-out-subject attribute transport, and useful defense claims require different splits and controls. This paper follows the attribute-leakage view, but the contribution is the claim-scoped audit: spectral attributes must be decodable above matched negative controls, and stronger identity or defense readings are blocked unless their own protocol cells support them.

\paragraph{EEG brainprints and biometric identity.}
EEG subject identifiability is a long-running observation: subject-specific physiology can act as a brainprint \citep{marcel2007person,palaniappan2007eegbiometrics,maiorana2015eegbiometrics,delpozo2015eegbiometric}. This work uses that fact as a boundary condition rather than as a shortcut. We do not optimize an EEG biometric classifier. A same-subject reference set can support identity linkage, but that does not imply identity recovery for held-out subjects. We therefore report reference-set $\bpid$ only for overlap splits and restrict the subject-disjoint claim to attribute leakage.

\paragraph{Neural data privacy and BCI security.}
Brain data raise privacy risks even when the released object is not a raw recording. BCI side-channel work shows that neural interfaces can reveal sensitive information through unintended channels \citep{martinovic2012sidechannel}, BCI platform discussions identify privacy and security risks in application-style neurotechnology pipelines \citep{bonaci2014app}, and neurotechnology ethics work frames neural data as a distinct privacy object \citep{ienca2017towards}. These works motivate the neural-data privacy problem. They do not provide the split-controlled representation-release audit studied here.

\paragraph{Membership inference and differential privacy.}
Membership inference and DP-SGD audits are natural tools for privacy evaluation \citep{shokri2017membership,yeom2018privacy,nasr2019comprehensive,carlini2022membership,abadi2016deep,dwork2006calibrating}, but they are not equivalent to representation privacy. Membership inference asks whether a training example was present; differential privacy constrains a training procedure; representation release asks what attributes remain decodable from a released vector. In our runs, attribute and reference-set identity leakage are much stronger than membership inference on the same embeddings. The training-noise experiments are treated as a stress test of the privacy-utility frontier, not as a formal DP guarantee.

\paragraph{Cross-dataset EEG audit setting.}
The empirical audit uses public EEG regimes with different acquisition assumptions: EEGMMI through PhysioNet and BCI2000 \citep{goldberger2000physionet,schalk2004bci2000}, Sleep-EDF through the PhysioNet sleep corpus \citep{kemp2000analysis}, LIMO through an 18-subject multi-channel EEG resource and its analysis ecosystem \citep{pernet2011limo,rousselet2016limo}, and CHB-MIT through a clinical pediatric scalp EEG corpus \citep{shoeb2009application}. Larger clinical EEG corpora such as TUH \citep{obeid2016temple} remain natural next targets, but the present claims are limited to the reported audit cells. These datasets are used to stress representation release under split controls, not to make clinical or diagnostic claims.

\paragraph{Concurrent 2025--2026 biosignal-FM privacy work.}
Recent biosignal foundation-model privacy work establishes that the question is timely without scooping the contribution here. Membership inference against ECG foundation encoders \citep{eegmia2025} shows that participation privacy can fail under embedding access, but only audits one endpoint and stays within ECG. Anonymized EEG synthesis with utility constraints \citep{eeganon2025} attacks the identity axis from the defense side under a synthesis pipeline, without joint MIA, attribute, or cross-encoder analysis on released embeddings. Memorization-risk audits for healthcare foundation models \citep{healthcarefm2025} adopt a black-box memorization frame on EHR data. EEG identity unlearning \citep{liu2024eegunlearn} and privacy-preserving cross-subject EEG transfer \citep{zhao2024eegtransferprivacy} address EEG identity protection from the defense side on classifier outputs rather than released foundation-model embeddings. None of these audits BIOT/LaBraM/EEGPT (or any EEG-FM) jointly across raw-copy, membership, identity, and attribute endpoints, defines an endpoint-disagreement object, or proves cross-encoder bridge transferability. Table~\ref{tab:related-delta} summarizes the delta against the four closest audit families.

\begin{table*}[!htbp]
\centering
\scriptsize
\setlength{\tabcolsep}{2.4pt}
\caption{Delta against the closest multi-endpoint audit families and concurrent 2025--2026 biosignal-FM privacy work. Y/N is whether the cited line of work \emph{reports} the corresponding endpoint or contribution as part of its own evaluation; ``S'' marks contributions that are scoped or partial.}
\label{tab:related-delta}
\begin{tabularx}{\textwidth}{@{}Xcccccc@{}}
\toprule
Audit family & Raw copy & MIA & Identity & Attribute & X-encoder transfer & Endpoint disagree\\
\midrule
Carlini-style memorization audit \citep{carlini2019secret,carlini2021extracting,carlini2022membership} & Y & Y & N & N & N & N \\
Embedding inversion / property leakage \citep{ateniese2015hacking,ganju2018property,song2020information,morris2023text} & S & N & N & Y & N & N \\
EEG biometric identity / unlearning \citep{marcel2007person,maiorana2015eegbiometrics,liu2024eegunlearn} & N & N & Y & N & N & N \\
ECG-FM MIA \citep{eegmia2025} & N & Y & N & N & N & N \\
EEG anonymization / re-id \citep{eeganon2025} & N & N & Y & N & N & N \\
\textbf{This paper} & \textbf{Y} & \textbf{Y} & \textbf{Y} & \textbf{Y} & \textbf{Y} & \textbf{Y} \\
\bottomrule
\end{tabularx}
\end{table*}

\paragraph{Position relative to prior work.}
The paper is not a raw inversion attack, an EEG biometric classifier, or a membership-inference / differential-privacy benchmark. Prior work supplies the ingredients: embeddings can leak attributes, EEG can identify enrolled users, and membership or DP answer different privacy questions. The contribution here is to force those ingredients into one EEG-FM release audit where each stronger interpretation is blocked unless its matching split or control supports it. The protocol tests random and split-permuted embeddings for attribute leakage, temporal-gap rows for adjacent-window artifacts, subject-disjoint splits for held-out-subject attribute claims, enrolled-subject reference sets for identity linkage, and utility gates for training-noise privacy claims.

\section{Threat Model and Audit Definitions}
\label{sec:threat}

\paragraph{Victim.}
The victim is an EEG foundation-model service or downstream system that releases embeddings, intermediate representations, or embedding-derived features. The raw EEG window is not released. This setting covers APIs that expose embeddings directly, internal representations shared with downstream partners, and applications that store embeddings for later personalization.

\paragraph{Attacker.}
The attacker has a shadow set of EEG windows and can train decoders from released embeddings to spectral attributes. For identity linkage, the attacker additionally has a reference set of candidate subjects. In subject-disjoint evaluation, train and test subjects do not overlap; identity labels in the test split are therefore not candidate labels for the attacker.

\paragraph{Attacker goals.}
The attacker has three possible goals. The first is attribute recovery: infer spectral attributes such as centered PSD structure from the embedding. The second is reference-set linkage: assign a target embedding or decoded attribute vector to a known subject in a reference set. The third is privacy auditing: test whether a release transformation or noisy training procedure reduces leakage without destroying task utility.

\paragraph{Evaluation principle.}
Each claim is tied to the split where it is meaningful. Attribute recovery is evaluated under window, temporal-gap, and subject-disjoint splits. Reference-set identity linkage is evaluated only when the candidate subjects appear in the reference set. Membership inference is evaluated as a train/test example audit and is not treated as a substitute for attribute leakage.

\paragraph{Out of scope.}
We do not claim raw waveform exfiltration, medical diagnosis inference, real-time attack deployment, or subject-disjoint identity recovery. We also do not claim formal differential privacy for the training-noise grids. These exclusions are part of the audit design because they prevent a strong split-controlled attribute result from being overstated as a stronger identity or defense result.

\begin{figure}[t]
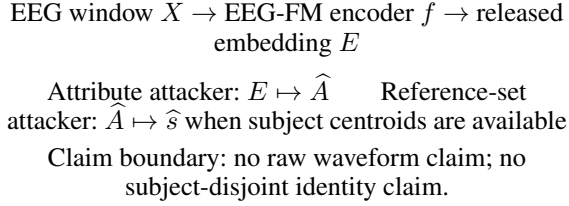

\centering
\fbox{\begin{minipage}{0.88\linewidth}
\centering
EEG window $X$ $\rightarrow$ EEG-FM encoder $f$ $\rightarrow$ released embedding $E$\\[2mm]
Attribute attacker: $E \mapsto \widehat{A}$ \hspace{5mm}
Reference-set attacker: $\widehat{A} \mapsto \widehat{s}$ when subject centroids are available\\[1mm]
Claim boundary: no raw waveform claim; no subject-disjoint identity claim.
\end{minipage}}
\caption{Threat model and claim boundary. The supported strict-split claim is spectral attribute leakage. Identity linkage is scoped to a reference-set setting.}
\label{fig:threat}
\end{figure}

\begin{definition}[Centered attribute leakage]
Let $X$ be an EEG window, $E=f(X)$ a released embedding, $A=a(X) \in \R^d$ a standardized spectral attribute vector, and $\mathcal{C}$ a set of negative-control embeddings such as random Gaussian embeddings and split-local permuted embeddings. For feature-centered random vectors $Y$ and $Z$, write
\begin{equation}
\operatorname{corr}_c(Y,Z)=\frac{\E\langle Y-\E Y,Z-\E Z\rangle}{\sqrt{\E\norm{Y-\E Y}^2\,\E\norm{Z-\E Z}^2}},
\end{equation}
with the empirical audit using the corresponding test-set plug-in estimate. For an attacker $g$, define
\begin{equation}
\attr(g) = \operatorname{corr}_c(A,g(E)) - \max_{C\in\mathcal{C}}\operatorname{corr}_c(A,g(C)).
\end{equation}
\end{definition}

\begin{definition}[Reference-set identity linkage]
Reference-set identity linkage is nearest-centroid or classifier-based subject identification when the candidate subject appears in the attacker's reference set. It is measured by $\bpid$. It is not evaluated in subject-disjoint splits because the candidate identities are absent.
\end{definition}

\begin{definition}[Privacy-utility gate]
For a defense parameter $\theta$, task utility $U(\theta)$, and privacy leakage $M(\theta)$, the tested defense passes the gate if $U(\theta) \geq u_0$ and $M(0)-M(\theta) \geq \Delta_M$ for at least one tested $\theta$.
\end{definition}

\section{Theory}
\label{sec:theory}

The formal object is the audit protocol, not a mechanistic model of BIOT, LaBraM, or EEGPT. The results below define which empirical statements are licensed by a split and its controls: when a centered attribute gain is evidence of representation dependence, when identity linkage is meaningful only because a reference set contains candidate subjects, why a noisy release can lose an identity-margin certificate while retaining attributes, and when a privacy-utility probe fails by definition. They do not certify privacy, prove raw waveform inversion, imply subject-disjoint identity recovery, or provide finite-sample generalization guarantees. Their role is claim control: a positive gain cannot be promoted across splits, subject-disjoint attribute transport cannot be rewritten as identity recovery, and a defense cannot be credited without satisfying both privacy and utility sides of the gate. This formal role is what lets the experiments report a strong attribute-leakage finding without upgrading it into a stronger identity or population-level privacy claim.

\paragraph{Notation.}
Let $X$ be an EEG window, $S$ its subject, $E=f(X)\in\R^p$ a released embedding, and $A=a(X)\in\R^d$ a standardized spectral attribute vector. In experiments $A$ is a PSD feature vector. All standardization maps are fit on the train split and then applied to the evaluation split. Let $\mathcal{C}_\pi$ denote the negative-control family for protocol cell $\pi$, such as random Gaussian embeddings or split-local permuted embeddings. For an attacker $g$ trained only on the train side of $\pi$, define the control-centered gain
\begin{equation}
\begin{aligned}
G_\pi(g)&=\operatorname{corr}_c(A_{\mathrm{test}},g(E_{\mathrm{test}}))\\
&\quad-\max_{C\in\mathcal{C}_\pi}\operatorname{corr}_c(A_{\mathrm{test}},g(C_{\mathrm{test}})).
\end{aligned}
\end{equation}
A positive value of $G_\pi(g)$ is the audit's centered attribute-gain witness.

\begin{definition}[Split-controlled attribute leakage]
A protocol cell $\pi$ supports split-controlled attribute leakage at level $\gamma>0$ if there exists an attacker $g$ trained under the split restrictions of $\pi$ such that $G_\pi(g)\geq\gamma$. The claim scope is inherited from the split: a window split supports same-subject attribute leakage, a temporal-gap split supports leakage not explained solely by adjacent-window copying, and a subject-disjoint split supports held-out-subject attribute leakage. None of these claims implies raw waveform reconstruction or identity recovery.
\end{definition}

\begin{proposition}[Control-centered gain witnesses representation dependence]
\label{prop:gain-witness}
Assume train-split standardization is fixed before evaluating the test split. All expectations in this proposition may be read either over the evaluation distribution or over the finite test-set empirical measure. Let $g$ be any attacker trained under protocol cell $\pi$. Suppose $A$ and $g(E)$ are feature-centered and evaluation-normalized with
\begin{gather}
\E\norm{A}^2/d=\E\norm{g(E)}^2/d=1,\\
\frac{1}{d}\E\norm{A-g(E)}^2\leq\epsilon,
\end{gather}
and suppose every negative control satisfies
\begin{equation}
\operatorname{corr}_c(A,g(C))\leq\rho_0
\quad\text{for all } C\in\mathcal{C}_\pi .
\end{equation}
Then
\begin{equation}
G_\pi(g)\geq 1-\epsilon/2-\rho_0.
\end{equation}
In particular, $G_\pi(g)>0$ whenever $\epsilon<2(1-\rho_0)$.
\end{proposition}

\begin{proof}
Under the normalization,
\begin{equation}
\frac{1}{d}\E\norm{A-g(E)}^2=2-2\operatorname{corr}_c(A,g(E)).
\end{equation}
The risk bound gives $\operatorname{corr}_c(A,g(E))\geq1-\epsilon/2$. Subtracting the largest control correlation gives the claim. The statement is conditional on the protocol cell and does not transfer the claim beyond that cell.
\end{proof}

\begin{definition}[Subject-disjoint attribute transport]
\label{def:transport}
Let $\pi_{\mathrm{sd}}$ be a subject-disjoint protocol cell with train subjects $\mathcal{S}_{\mathrm{tr}}$ and test subjects $\mathcal{S}_{\mathrm{te}}$, where $\mathcal{S}_{\mathrm{tr}}\cap\mathcal{S}_{\mathrm{te}}=\varnothing$. An embedding satisfies $(\epsilon_{\mathrm{tr}},\epsilon_{\mathrm{te}})$ attribute transport for decoder class $\mathcal{G}$ if there exists $g\in\mathcal{G}$ trained on $\mathcal{S}_{\mathrm{tr}}$ such that
\begin{equation}
\begin{aligned}
&\frac{1}{d}\E\!\left[\norm{A-g(E)}^2\mid S\in\mathcal{S}_{\mathrm{tr}}\right]
\leq\epsilon_{\mathrm{tr}},\\
&\frac{1}{d}\E\!\left[\norm{A-g(E)}^2\mid S\in\mathcal{S}_{\mathrm{te}}\right]
\leq\epsilon_{\mathrm{te}} .
\end{aligned}
\end{equation}
The second inequality is the non-vacuous part: it states that the learned attribute map transports to held-out subjects. It is tested by subject-disjoint evaluation, not implied by the window split.
\end{definition}

\begin{theorem}[Protocol-level separation of attribute leakage and a sufficient identity-margin certificate]
\label{thm:protocol-separation}
Fix a protocol cell $\pi$ and decoder $h_\sigma$ whose negative-control ceiling, evaluated under the same noisy-release protocol and decoder, satisfies $\operatorname{corr}_c(A,h_\sigma(C))\leq\rho_0$ for every $C\in\mathcal{C}_\pi$. Let $A\in\R^d$ be centered with $\E A=0$ and $\E\norm{A}^2=d$, and assume the decoded noise term below has mean zero so centered and uncentered second-moment expressions coincide. Suppose the released embedding decomposes as
\begin{equation}
E = U A + H,
\qquad
U^\top U=\alpha I_d,
\qquad
U^\top H=0,
\end{equation}
where $\alpha>0$ measures the strength of an attribute-bearing subspace and $H$ contains components orthogonal to that subspace. For Gaussian release noise
\begin{equation}
E_\sigma=E+Z,
\qquad
Z\sim\mathcal{N}(0,\sigma^2 I_p),
\end{equation}
independent of $(A,H)$, define $h_\sigma(E_\sigma)=\alpha^{-1}U^\top E_\sigma$. Then
\begin{equation}
\operatorname{corr}_c(A,h_\sigma(E_\sigma))=\left(1+\frac{\sigma^2}{\alpha}\right)^{-1/2}
\end{equation}
and therefore
\begin{equation}
G_\pi(h_\sigma)\geq\left(1+\frac{\sigma^2}{\alpha}\right)^{-1/2}-\rho_0.
\end{equation}
Thus, for any $\gamma>0$ with $0<\gamma+\rho_0<1$, attribute leakage remains certified at level $\gamma$ whenever
\begin{equation}
\sigma/\sqrt{\alpha}<\sqrt{\frac{1}{(\gamma+\rho_0)^2}-1}.
\end{equation}

Now suppose the same decoded attributes are used for reference-set nearest-centroid identity linkage. Let $\mu_s$ be the reference centroid for enrolled subject $s$, let all clean within-subject deviations satisfy $\norm{A-\mu_S}\leq r$, and let $\Delta=\min_{s\neq s'}\norm{\mu_s-\mu_{s'}}$. A sufficient clean identity margin is $m=\Delta/2-r>0$. Under the noisy decoded release, the added attribute-space perturbation is $\mathcal{N}(0,\sigma^2\alpha^{-1}I_d)$. Hence a standard Gaussian-norm certificate guarantees preservation of the nearest-centroid identity decision with probability at least $1-\delta$ whenever
\begin{equation}
\frac{\sigma}{\sqrt{\alpha}}\left(\sqrt d+\sqrt{2\log(1/\delta)}\right)<m.
\end{equation}
Consequently, if for some $\gamma>0$ with $0<\gamma+\rho_0<1$ the non-empty-regime condition
\begin{equation}
\frac{m}{\sqrt d+\sqrt{2\log(1/\delta)}}
<
\sqrt{\frac{1}{(\gamma+\rho_0)^2}-1}
\end{equation}
holds and
\begin{equation}
\frac{m}{\sqrt d+\sqrt{2\log(1/\delta)}}
\leq
\frac{\sigma}{\sqrt{\alpha}}
<
\sqrt{\frac{1}{(\gamma+\rho_0)^2}-1},
\end{equation}
then this margin certificate no longer certifies reference-set identity at confidence $1-\delta$, while the centered attribute-gain lower bound remains at least $\gamma$.
\end{theorem}

\begin{proof}
The decoder satisfies
\begin{equation}
h_\sigma(E_\sigma)=\alpha^{-1}U^\top(UA+H+Z)=A+\alpha^{-1}U^\top Z,
\end{equation}
because $U^\top U=\alpha I_d$ and $U^\top H=0$. The perturbation term is Gaussian with covariance $\sigma^2\alpha^{-1}I_d$ and is independent of $A$. Therefore
\begin{equation}
\E\langle A,h_\sigma(E_\sigma)\rangle=d,
\qquad
\E\norm{h_\sigma(E_\sigma)}^2=d\left(1+\frac{\sigma^2}{\alpha}\right),
\end{equation}
which gives the stated centered correlation. Subtracting the control ceiling $\rho_0$ gives the gain bound and the displayed condition for gain at least $\gamma$.

For identity linkage, nearest-centroid recovery is guaranteed whenever the decoded perturbation has norm below the clean margin $m=\Delta/2-r$. If $G\sim\mathcal{N}(0,I_d)$, then $\norm{G}\leq\sqrt d+\sqrt{2\log(1/\delta)}$ with probability at least $1-\delta$. Scaling by $\sigma/\sqrt{\alpha}$ gives the sufficient certificate. When the lower inequality holds, that sufficient certificate is no longer available; when the upper inequality holds, the attribute-gain lower bound remains at least $\gamma$. The conclusion is a separation between a sufficient reference-set identity certificate and a positive attribute-leakage certificate, not a proof that identity linkage is impossible.
\end{proof}

\begin{remark}[How to read the separation]
Theorem~\ref{thm:protocol-separation} does not say Gaussian noise is an adequate privacy mechanism, nor that attribute leakage always survives arbitrary defenses. It is a stylized sufficient-condition result: under an orthogonal redundant attribute subspace with isotropic release noise and a nearest-centroid margin model, there is a parameter regime in which a margin-based reference-set identity certificate is lost while attribute leakage remains certified above random or permuted controls. The defense curves test for the same qualitative failure mode, but Theorem~\ref{thm:protocol-separation} is not a fitted model of BIOT or LaBraM embeddings.
\end{remark}

\begin{theorem}[Endpoint-independence: bounded membership and identity certificates do not imply bounded attribute leakage]
\label{thm:endpoint-independence}
Let the released embedding follow the redundant-attribute model of Theorem~\ref{thm:protocol-separation}: $E=UA+H$ with $U^\top H=0$, $U^\top U=\alpha I_d$, and $E_\sigma=E+Z$ where $Z\sim\mathcal{N}(0,\sigma^2I_p)$ is independent of $(A,H)$. Let $h_\sigma(E_\sigma)=\alpha^{-1}U^\top E_\sigma$ be the matched attribute decoder. Suppose the noisy-release attribute attacker is audited against any negative-control family with ceiling $\rho_0\in[0,1)$. Suppose further that on the audited subject-disjoint split, every classifier from $E_\sigma$ to the membership label of the released encoder has AUC at most $\tfrac12+\delta_M$ for some $\delta_M\in(0,\tfrac12]$ (this is an audited assumption: we report this AUC empirically, and the theorem only quantifies its consequences when the assumption holds). Fix any required attribute-leakage threshold $\gamma\in(0,1-\rho_0)$ and any reference-set identity confidence $1-\delta_I\in(0,1)$ together with a clean nearest-centroid margin model with margin parameter $m$ and attribute dimension $d$. Then there exists a parameter family $(\alpha,\sigma,m,d)$ in which all three of the following hold simultaneously on the same released embeddings:
\begin{enumerate}
    \item the assumed membership-AUC bound $\tfrac12+\delta_M$ is preserved (by hypothesis);
    \item the sufficient nearest-centroid identity certificate of Theorem~\ref{thm:protocol-separation} fails at confidence $1-\delta_I$, because $\sigma/\sqrt{\alpha}\geq m/(\sqrt d+\sqrt{2\log(1/\delta_I)})$;
    \item the centered attribute-leakage gain satisfies $G_\pi(h_\sigma)\geq\gamma$, because $\sigma/\sqrt{\alpha}<\sqrt{(\gamma+\rho_0)^{-2}-1}$.
\end{enumerate}
The membership and identity bounds therefore do not constrain the attribute-leakage bound: any tightening on the membership/identity side leaves the attribute side free in this parameter family.
\end{theorem}

\begin{proof}
Items (2) and (3) are direct rewrites of Theorem~\ref{thm:protocol-separation} once $\sigma/\sqrt{\alpha}$ is placed in the open interval
\begin{equation}
\Bigl[\,\frac{m}{\sqrt d+\sqrt{2\log(1/\delta_I)}},\ \sqrt{(\gamma+\rho_0)^{-2}-1}\,\Bigr).
\end{equation}
The non-empty-regime condition of Theorem~\ref{thm:protocol-separation} guarantees that this interval is non-empty for some choice of $(\alpha,\sigma,m,d)$, and any such choice realizes (2) and (3) on the same released embeddings. Item (1) is the audited hypothesis: in the audited cell the released encoder is frozen, all attribute decoders are trained on a shadow set drawn from the same window distribution as the test split, and the auditor reports the empirical membership AUC on a held-out subject-disjoint split. The theorem makes no claim that this AUC is automatically $\tfrac12$; it states that whenever the audit measures AUC at most $\tfrac12+\delta_M$, the parameter family above realizes (1)--(3) jointly on the same embeddings.
\end{proof}

\begin{remark}[How to read endpoint-independence]
Theorem~\ref{thm:endpoint-independence} is a parameter-existence statement: it shows that bounding membership AUC and the sufficient identity certificate cannot in general bound attribute leakage, by exhibiting a parameter family of released embeddings for which all three statements hold jointly. It is not a property of BIOT, LaBraM, or EEGPT, and we do not claim that real EEG-FM embeddings sit in this parameter family. The role of the theorem is to forbid a deductive shortcut from a single endpoint to a privacy verdict; the empirical disagreement reported in Table~\ref{tab:audit-comparison} and the AEDS margins in Table~\ref{tab:aeds} are what tie the existence statement to the audited encoders.
\end{remark}

\begin{definition}[Audit-endpoint disagreement]
\label{def:aeds}
Fix an audit cell $\pi$. Let $L^A_\pi$ be a lower confidence bound or theoretical lower bound on the matched-control attribute gain. Let
\begin{equation}
R^M_\pi=\max\{0,\operatorname{AUC}_{\mi,\pi}-1/2\}
\end{equation}
be the membership residual when a utility-valid membership endpoint is reported under the same protocol cell, and let $R^I_\pi$ be the certified identity residual under the same claim scope. If no utility-valid membership endpoint is available, or if a subject-disjoint split has no candidate identities and therefore no identity certificate, the corresponding residual is zero. Each endpoint is mapped to slack above its no-risk or uncertified baseline; the three slack quantities live in $[0,1]$ but are not in a common physiological unit. Cell $\pi$ is endpoint-disagreeing when
\begin{equation}
L^A_\pi>\max\{R^M_\pi,R^I_\pi\},
\end{equation}
in which case the audit attribute endpoint certifies positive residual evidence that no other audited endpoint certifies for the same cell. The companion margin
\begin{equation}
\aeds_\pi=L^A_\pi-\max\{R^M_\pi,R^I_\pi\}
\end{equation}
quantifies how far the attribute slack exceeds the strongest other endpoint slack at the same cell; we treat $\aeds_\pi$ as a decision-margin diagnostic, not a common-unit privacy score, and report $L^A_\pi$, $R^M_\pi$, $R^I_\pi$ separately in every reported cell.
\end{definition}

\begin{theorem}[Positive AEDS under endpoint independence]
\label{thm:aeds-positive}
Use the redundant-attribute noisy-release model of Theorem~\ref{thm:protocol-separation}. Let the negative-control ceiling be at most $\rho_0$, and suppose a utility-valid membership endpoint has residual at most $\delta_M$. Fix an AEDS margin $\eta>0$ with $\eta+\delta_M+\rho_0<1$, identity confidence $1-\delta_I$, identity margin $m$, and $c_I=\sqrt d+\sqrt{2\log(1/\delta_I)}$. If
\begin{equation}
\frac{m}{c_I}<\sqrt{\frac{1}{(\eta+\delta_M+\rho_0)^2}-1}
\end{equation}
and the noise-to-attribute ratio satisfies
\begin{equation}
\frac{m}{c_I}\leq \frac{\sigma}{\sqrt{\alpha}}<\sqrt{\frac{1}{(\eta+\delta_M+\rho_0)^2}-1},
\end{equation}
then the sufficient reference-set identity certificate is unavailable at confidence $1-\delta_I$, while
\begin{equation}
\aeds_\pi\geq\eta .
\end{equation}
Thus the endpoint-independence interval contains a sub-interval where the attribute audit and the membership/identity endpoint audit make opposite deployment decisions with positive margin.
\end{theorem}

\begin{proof}
The lower bound in Theorem~\ref{thm:protocol-separation} gives
\begin{equation}
L^A_\pi\geq\left(1+\frac{\sigma^2}{\alpha}\right)^{-1/2}-\rho_0.
\end{equation}
The upper bound on $\sigma/\sqrt{\alpha}$ implies $L^A_\pi>\eta+\delta_M$. The lower bound on $\sigma/\sqrt{\alpha}$ is exactly the condition under which the stated nearest-centroid identity certificate is no longer available, so $R^I_\pi=0$ for the same claim scope. With $R^M_\pi\leq\delta_M$, Definition~\ref{def:aeds} gives $\aeds_\pi\geq L^A_\pi-\delta_M\geq\eta$.
\end{proof}

\begin{proposition}[Positive subject-disjoint gain supports attributes, not held-out-subject identity]
\label{prop:subject-disjoint-scope}
Consider a subject-disjoint protocol cell with train subjects $\mathcal{S}_{\mathrm{tr}}$ and test subjects $\mathcal{S}_{\mathrm{te}}$, where $\mathcal{S}_{\mathrm{tr}}\cap\mathcal{S}_{\mathrm{te}}=\varnothing$. Under the centering and evaluation-normalization assumptions of Proposition~\ref{prop:gain-witness}, if a decoder trained on $\mathcal{S}_{\mathrm{tr}}$ satisfies the test-side condition in Definition~\ref{def:transport} and the negative-control ceiling on $\mathcal{S}_{\mathrm{te}}$ is at most $\rho_0$, then
\begin{equation}
G_{\mathrm{sd}}(g)\geq1-\epsilon_{\mathrm{te}}/2-\rho_0 .
\end{equation}
This lower bound is a positive leakage certificate only when $\epsilon_{\mathrm{te}}<2(1-\rho_0)$. However, this protocol cell does not define reference-set identity linkage for the held-out subjects unless centroids or labels for the candidate test identities are available to the attacker.
\end{proposition}

\begin{proof}
The gain lower bound is Proposition~\ref{prop:gain-witness} applied to the held-out-subject test distribution. The identity statement follows from the definition of reference-set linkage: nearest-centroid or classifier-based identity assignment requires the candidate identities to be present in the attacker's reference set. Subject-disjoint evaluation removes those candidates from the training/reference side, so a positive attribute-gain result is not an identity-recovery result.
\end{proof}

\begin{theorem}[Bridge transferability under shared attribute subspaces]
\label{thm:bridge-transfer}
Consider a subject-disjoint protocol cell with paired source- and target-encoder embeddings generated from the same windows:
\begin{equation}
E_s=U_sP_sA+H_s,\qquad E_t=U_tP_tA+H_t+Z_t,
\end{equation}
where $A\in\R^d$ is isotropic and centered with $\E\norm{A}^2=d$, $P_s$ and $P_t$ are orthogonal projectors in attribute-coordinate space, $U_s^\top U_s=U_t^\top U_t=I_d$, and $U_t^\top H_t=0$. Assume $W_t=U_t^\top Z_t$ is independent of $A$ with covariance $\tau_t^2I_d$, and define
\begin{equation}
\beta=\frac{1}{d}\norm{P_sP_t}_F^2.
\end{equation}
Let $h_s(y)=U_s^\top y$ be the source-space linear attribute decoder, and let $B_\lambda$ be a fitted ridge bridge from target embeddings to source embeddings. Suppose the fitted bridge has centered correlation at most $\epsilon_{\mathrm{br}}$ below the shared-projection bridge $B_\cap=U_sP_sP_tU_t^\top$ under the same evaluation normalization. If the target negative-control ceiling for the chained attacker is at most $\rho_0$, then
\begin{equation}
G_{\mathrm{sd}}(h_s\circ B_\lambda)
\geq
\sqrt{\frac{\beta}{1+\tau_t^2}}-\epsilon_{\mathrm{br}}-\rho_0 .
\end{equation}
Thus a nontrivial overlap between the source- and target-exposed attribute projectors gives a positive cross-encoder attribute-gain certificate whenever $\sqrt{\beta/(1+\tau_t^2)}>\rho_0+\epsilon_{\mathrm{br}}$.
\end{theorem}

\begin{proof}
For the shared-projection bridge, the chained decoder is
\begin{equation}
h_s(B_\cap E_t)=P_sP_tA+P_sP_tW_t.
\end{equation}
Let $K=P_sP_t$. For orthogonal projectors, $\operatorname{tr}(K)=\norm{P_sP_t}_F^2=\beta d$ and $\operatorname{tr}(KK^\top)=\norm{P_sP_t}_F^2=\beta d$. Since $A$ is isotropic and $W_t$ is independent isotropic noise in the target attribute coordinates,
\begin{equation}
\begin{aligned}
\operatorname{corr}_c\!\bigl(A,K(A+W_t)\bigr) &= \frac{\operatorname{tr}(K)}{\sqrt{d(1+\tau_t^2)\operatorname{tr}(KK^\top)}} \\
&= \sqrt{\beta/(1+\tau_t^2)}.
\end{aligned}
\end{equation}
The fitted bridge loses at most $\epsilon_{\mathrm{br}}$ centered correlation relative to this bridge. Subtracting the target-control ceiling gives the displayed gain bound.
\end{proof}

\begin{theorem}[Bottleneck dimension alone is not an attribute privacy certificate]
\label{prop:bottleneck}
Let $E=UA+H$ with $U^\top H=0$ and $U^\top U=\alpha I_d$ as in Theorem~\ref{thm:protocol-separation}. Let $B\in\R^{k\times p}$ be a release bottleneck. For $k\geq d$, if $BU$ has full column rank, then the attribute-bearing component of the bottlenecked release $BE$ is linearly identifiable up to the inverse smallest singular value of $BU$: the linear decoder
\begin{equation}
q(BE)=((BU)^\top BU)^{-1}(BU)^\top BE
\end{equation}
satisfies $q(BE)=A$ when $BH$ is orthogonal to the column space of $BU$. More generally, if the normalized mean-square projection obeys $d^{-1}\E\norm{P_{\mathrm{col}(BU)}BH}^2\leq\tau^2$, then
\begin{equation}
\frac{1}{d}\E\norm{A-q(BE)}^2\leq
\tau^2\left\|((BU)^\top BU)^{-1}(BU)^\top\right\|_{\mathrm{op}}^2.
\end{equation}
Equivalently, with the centering and evaluation-normalization assumptions of Proposition~\ref{prop:gain-witness} and a control ceiling $\rho_0$, the bottlenecked release admits a centered-gain lower bound
\begin{equation}
G_\pi(q\circ B)\geq 1-\tfrac12\tau^2\left\|((BU)^\top BU)^{-1}(BU)^\top\right\|_{\mathrm{op}}^2-\rho_0,
\end{equation}
which is positive whenever $\tau^2\left\|((BU)^\top BU)^{-1}(BU)^\top\right\|_{\mathrm{op}}^2<2(1-\rho_0)$. The bound is invariant to the bottleneck width $k$ once $k\geq d$ and $BU$ has full column rank, so a small $k$ alone does not certify attribute privacy. When $k<d$, this theorem does not apply without an additional lower-dimensional attribute model.
\end{theorem}

\begin{proof}
Substituting $BE=BUA+BH$ into $q$ gives $q(BE)-A=((BU)^\top BU)^{-1}(BU)^\top BH$. The equality case follows when $BH$ has zero projection onto the column space of $BU$. The mean-square bound follows by applying the operator norm to this error term. The centered-gain restatement follows by combining the normalization of Proposition~\ref{prop:gain-witness} with $d^{-1}\E\norm{A-q(BE)}^2\leq\tau^2\left\|((BU)^\top BU)^{-1}(BU)^\top\right\|_{\mathrm{op}}^2$ and subtracting $\rho_0$. The width-invariance statement follows because the bound depends on $B$ only through $((BU)^\top BU)^{-1}(BU)^\top$ and $\tau$, neither of which is forced to grow with $k$.
\end{proof}

\begin{proposition}[Negative privacy-utility frontier]
\label{prop:frontier}
Let $U(\theta)$ be task utility and $M(\theta)$ privacy leakage under a tested defense parameter $\theta$. For utility threshold $u_0$ and required leakage reduction $\Delta_M$, define $\mathcal{S}_u=\{\theta:U(\theta)\geq u_0\}$, with $\sup\varnothing=-\infty$. If
\begin{equation}
\sup_{\theta\in\mathcal{S}_u}\left[M(0)-M(\theta)\right]<\Delta_M,
\end{equation}
then no tested defense setting satisfies the privacy-utility gate.
\end{proposition}

\begin{proof}
The gate requires a tested parameter $\theta$ satisfying both $\theta\in\mathcal{S}_u$ and $M(0)-M(\theta)\geq\Delta_M$. The displayed inequality rules out such a parameter; if $\mathcal{S}_u$ is empty, the utility requirement already fails.
\end{proof}

Together, these results define an audit-claim-control framework, not a mechanistic theorem about EEG-FMs. Proposition~\ref{prop:gain-witness} licenses a positive centered gain as attribute leakage only relative to the evaluated split, decoder, normalization, and controls; the claim cannot be promoted across splits. Theorem~\ref{thm:protocol-separation} gives a stylized regime in which a sufficient reference-set identity certificate can fail while attribute leakage remains positive, and Theorem~\ref{thm:endpoint-independence} shows that bounded membership AUC and a failed identity certificate do not imply bounded attribute leakage on the same released embeddings. Definition~\ref{def:aeds} turns this endpoint disagreement into a scalar decision-margin diagnostic, and Theorem~\ref{thm:aeds-positive} gives a sufficient interval in which that margin is positive. Proposition~\ref{prop:subject-disjoint-scope} shows why a positive subject-disjoint gain is an attribute statement and cannot be rewritten as held-out identity recovery. Theorem~\ref{thm:bridge-transfer} gives the cross-encoder analogue: if two encoders expose a shared attribute projection and a fitted bridge is close to the shared-projection bridge, transfer leakage has a positive lower bound after controls. Theorem~\ref{prop:bottleneck} sharpens the defense side by giving a quantitative bottleneck lower bound that is invariant to bottleneck width once $k\geq d$ and the attribute subspace is preserved. Proposition~\ref{prop:frontier} blocks crediting any defense that does not satisfy both the privacy-reduction and the utility sides of the gate. The framework is portable: any released representation paired with an attribute target, an attacker class, a control family, and a utility metric can be audited by the same protocol cells, regardless of whether the underlying encoder is an EEG foundation model. These statements support the audit protocol; they do not provide a universal impossibility result, a certified privacy mechanism, or a fitted model of the empirical curves.

\section{Experimental Setup}
\label{sec:setup}

\paragraph{Model and data.}
We evaluate BIOT, LaBraM, and EEGPT embeddings on PhysioNet EEGMMI R03/R07 windows, LaBraM embeddings on Sleep-EDF windows, LIMO EEG epochs, and CHB-MIT windows, and T1/T2 event segments for the privacy-utility probe. The main representation audit uses BIOT EEG-PREST-16 R03 embeddings over 20 subjects. A checkpoint-robustness audit adds a BIOT SHHS+PREST-18 checkpoint on the same 20-subject EEGMMI R03 windows, using the checkpoint's 18-channel token table with the available 16 EEGMMI channels. Cross-architecture audits add LaBraM-base and EEGPT MCAE R03 embeddings on the same 20-subject, 16-channel, 2000-sample window protocol. Cross-dataset audits add LaBraM-base Sleep-EDF age evaluations with 10 and 13 subjects, Fpz-Cz/Pz-Oz EEG channels, and 2000-sample windows; an 18-subject LIMO evaluation with 54 mapped 10--20 channels and 200-sample EEG epochs; and a 23-case CHB-MIT pediatric scalp EEG evaluation with 16 bipolar montage channels and 2000-sample windows. An acquisition-robustness audit repeats the 20-subject window protocol on EEGMMI R07 with BIOT EEG-PREST-16. The privacy-utility audit uses a 60-subject event evaluation with 900 T1/T2 trials. All reported headline values are regenerated from the archived validation tables.

\paragraph{Splits.}
We use three representation splits. The window split evaluates same-subject train/test windows. The temporal-gap split separates windows by a temporal gap to reduce adjacent-window artifacts; because that row uses random controls rather than the full random/permuted matched-control set, we treat it as a sanity check rather than a primary replication cell. The subject-disjoint split uses non-overlapping train and test subjects and supports only attribute leakage claims. The subject-head privacy-utility proxy uses same-subject, index-disjoint train/test windows; the event-level task split evaluates T1/T2 membership and utility frontiers with index-disjoint train/test events.

\paragraph{Attacks.}
The attribute attackers include ridge, MLP, and kNN decoders from embeddings to centered PSD attributes. Each decoder is trained on standardized train-split PSD features and evaluated by centered correlation on the test split. Reference-set identity is measured by $\bpid$ over decoded attributes or embedding-space nearest centroids when candidate subjects appear in the reference set. Membership inference uses loss or confidence on train/test examples. The training-noise grids should be read as privacy-utility probes, not formal DP accounting.

\paragraph{Controls and gates.}
Every headline leakage result is compared against the available negative-control embeddings, with random and split-local permuted controls for the main window/subject-disjoint and later sweep rows, and random controls for the temporal-gap row. The operational replication rule used for attack/baseline sweeps requires window MLP gain above $0.2$, at least one simpler ridge or kNN baseline above $0.05$, random/permuted MLP controls with absolute centered correlation at most $0.15$, and subject-disjoint MLP gain above $0.1$ for held-out-subject attribute claims. The temporal-gap row is used as an adjacent-window sanity check rather than as the primary replication cell because its temporal-gap row uses random controls only. The seed-matched confidence-interval audit pairs each strong residual MLP row with random, split-permuted, subject-block-permuted, within-subject-roll, and target-permuted controls from the same model, split, and seed. Its main negative-control gain uses the stronger same-seed value among random and split-local permuted embeddings. The audit includes finite-value checks, duplicate-row checks, temporal-gap controls, subject-disjoint controls, raw-window near-duplicate checks, and defense monotonicity checks. PSD scaling uses train-set statistics; centered correlation is an evaluation metric computed on the evaluation split. Table values are reported as means over seeded runs when available.

\paragraph{Statistical reporting.}
Unless otherwise stated, $\pm$ values are sample standard deviations across the five seeded split runs. Confidence intervals in Table~\ref{tab:matched-ci} are two-sided 95\% $t$ intervals over the same five seeds. The unit of uncertainty is the seeded train/test split, not an individual EEG window. Matched-control contrasts are computed within each seed before aggregation, so a gain compares the strong residual MLP to controls trained and evaluated on the same model, split, and random seed. Because each interval has only five seeded splits, we use these intervals as stability screens and matched-control summaries rather than precise uncertainty estimates. These intervals are descriptive finite-split audit summaries; they are not population-level guarantees over all subjects, devices, or EEG datasets.

\paragraph{Reproducibility records.}
Table~\ref{tab:protocol} summarizes split and provenance details for the main empirical cells, and Table~\ref{tab:repro} gives the minimum reproduction metadata for the representation audit. The accompanying artifact release contains the aggregate leakage table, split-leakage figure data, defense-curve data, privacy-utility figure data, threshold-sensitivity tables, full random/permuted defense controls, LaBraM/EEGPT defense-control curves, the seed-matched high-capacity attacker confidence-interval summary, the AEDS endpoint-disagreement table with paired-seed bootstrap and per-cell head-level LiRA, the cross-encoder transfer summary, the bridge-theorem $\beta$ consistency-check columns, the deployment-grade attribute-classification table, the adaptive-attacker sweep, the DP-SGD attribute-audit sweep, and a manifest mapping every reported table or figure to its validation record. The validation pipeline is organized by phase: Phase~24 produces the ridge/MLP/kNN attack-baseline table; Phase~25 produces the defense-control and gate-sensitivity records; Phase~26 produces the high-capacity matched-control sweeps; Phase~27 produces Table~\ref{tab:matched-ci} and Figure~\ref{fig:sd-ci}; Phase~28 supplies the utility-valid membership residual used by AEDS and produces the balanced Sleep-EDF privacy-utility frontier; Phase~29 produces Table~\ref{tab:cross-encoder-transfer}; Phase~30 produces the deployment-grade attribute classification numbers in Section~\ref{sec:leakage} and the adaptive-attacker sweep in Table~\ref{tab:adaptive-attacker}; Phase~31 produces the bridge-theorem $\beta$ columns of Table~\ref{tab:cross-encoder-transfer} and the paired-seed bootstrap calibration of Table~\ref{tab:aeds}; Phase~32 produces the per-cell head-level LiRA $R^M$ values in Table~\ref{tab:aeds}; Phase~34 produces Table~\ref{tab:dp-sgd}. The shared standardization helper fits train-split mean and standard deviation and applies them to both train and test arrays; PSD decoder targets are standardized with train statistics; centered correlation is used only as an evaluation metric; subject-disjoint splits raise an error on subject overlap and report overlap zero in the result table; the subject-head privacy-utility probe uses same-subject index-disjoint train/test windows and is therefore described only as a subject-head proxy.

\begin{table}[!htbp]
\centering
\tiny
\setlength{\tabcolsep}{2pt}
\renewcommand{\arraystretch}{0.92}
\caption{Protocol summary for the empirical audit cells used in the headline results. Standardization is fit on the train split and applied to test data for all representation attacks. Subject-disjoint identity is not evaluated because candidate identities are absent from the reference set.}
\label{tab:protocol}
\begin{tabularx}{\linewidth}{@{}Xccccl@{}}
\toprule
Cell & Subj. & Train/test & Seeds & Controls & Scope \\
\midrule
R03 PREST-16 window & 20 & 260/140 win & 5 & rand, perm & Attr.+ID \\
R03 PREST-16 temporal gap & 20 & 220/140 win & 5 & rand & Attr.+ID(gap) \\
R03 PREST-16 subj.-disjoint & 20 & 13/7 subj & 5 & rand, perm & Attr.\\
R03 SHHS+PREST-18 window & 20 & 260/140 win & 5 & rand, perm & Attr.+ID \\
R03 SHHS+PREST-18 subj.-disjoint & 20 & 13/7 subj & 5 & rand, perm & Attr.\\
R03 LaBraM-base window & 20 & 260/140 win & 5 & rand, perm & Attr.+ID \\
R03 LaBraM-base subj.-disjoint & 20 & 13/7 subj & 5 & rand, perm & Attr.\\
R03 EEGPT-MCAE window & 20 & 260/140 win & 5 & rand, perm & Attr.+ID \\
R03 EEGPT-MCAE subj.-disjoint & 20 & 13/7 subj & 5 & rand, perm & Attr.\\
R07 PREST-16 window & 20 & 260/140 win & 5 & rand, perm & Attr.+ID \\
R07 PREST-16 subj.-disjoint & 20 & 13/7 subj & 5 & rand, perm & Attr.\\
Sleep-EDF LaBraM window & 13 & 169/91 win & 5 & rand, perm & Attr.; weak ID \\
Sleep-EDF LaBraM subj.-disjoint & 13 & 8/5 subj & 5 & rand, perm & Attr.\\
LIMO LaBraM window & 18 & 702/378 ep & 5 & rand, perm & Attr.+ID \\
LIMO LaBraM subj.-disjoint & 18 & 12/6 subj & 5 & rand, perm & Attr.\\
CHB-MIT LaBraM window & 23 & 299/161 win & 5 & rand, perm & Attr.+ID \\
CHB-MIT LaBraM subj.-disjoint & 23 & 15/8 cases & 5 & rand, perm & Attr.\\
Event privacy-utility probe & 60 & idx-disjoint ev & 5 & noise grid & Neg.\ frontier \\
\bottomrule
\end{tabularx}
\end{table}

\begin{table}[!htbp]
\centering
\tiny
\setlength{\tabcolsep}{2.6pt}
\renewcommand{\arraystretch}{0.92}
\caption{Reproducibility metadata for the representation audit. The archived validation tables regenerate the reported result tables.}
\label{tab:repro}
\begin{tabularx}{\linewidth}{@{}llX@{}}
\toprule
Component & Setting & Reproduction detail \\
\midrule
Raw windows & EEGMMI R03/R07 & 16 selected EEG channels, 2000 samples per window, per-window channel mean removal and SD scaling \\
LaBraM/EEGPT windows & EEGMMI R03 & Same 20-subject, 16-channel, 2000-sample window protocol; EDF channel names mapped to each architecture's channel table \\
Cross-dataset & SE/LIMO/CHB & SE: Fpz-Cz/Pz-Oz $\to$ LaBraM frontal/parietal (2000 samples); LIMO: 54 mapped 10--20 channels (200 samples); CHB-MIT: 16 bipolar montage channels (2000 samples) \\
Embeddings & BIOT/LaBraM/EEGPT & Frozen pretrained encoders; archived records include embeddings, raw-window provenance, subject labels, splits, model family, checkpoint \\
Attribute target & PSD features & Computed from audited raw windows; target standardization fit on train, applied to test \\
Splits & win/gap/subj.-disjoint & Five seeds for attack/baseline sweeps; subj.-disjoint enforces zero overlap, no identity evaluated \\
Attackers & ridge/MLP/kNN & Ridge fixed $\ell_2$; MLP defaults; kNN $k{=}5$ PSD-neighbor averaging \\
Controls & random/permuted & Random Gaussian and split-local deranged row permutations; temporal-gap canonical row uses random controls \\
Archived tables & validation records & Aggregate leakage, defense, privacy-utility, threshold-sensitivity, matched-control, reproduction manifest \\
\bottomrule
\end{tabularx}
\end{table}

\section{Representation Leakage Results}
\label{sec:leakage}

\subsection{Spectral attributes are recoverable from embeddings}

The primary evidence unit is the seed-matched high-capacity audit reported in Table~\ref{tab:matched-ci} and visualized in Figure~\ref{fig:sd-ci}. Across eight evaluation settings, every subject-disjoint 95\% interval remains above zero against the stronger same-seed random/split-permuted negative control and against target permutation. The weakest lower bounds are $0.174$ for matched negative controls and $0.157$ for target permutation; the LIMO subject-disjoint lower bounds are $0.275$ and $0.257$, respectively, and the CHB-MIT pediatric scalp EEG subject-disjoint lower bounds are $0.374$ and $0.346$. Within-seed pairing means each row compares the strong residual MLP to controls trained and evaluated on the same model, split, and seed, which removes the seed-variance and control-quality confounds that flat reruns leave open. These rows preserve the claim boundary: same-subject rows can support reference-set identity when candidate subjects are present, while subject-disjoint rows support attribute leakage only.

The matched-CI finding anchors a larger endpoint-disagreement result. On the same released embeddings, raw-copy, membership-only, and same-subject reference-set audits return verdicts that point in different deployment directions, while only the split-controlled attribute audit licenses the held-out-subject leakage claim above. Table~\ref{tab:audit-comparison} in Section~\ref{sec:protocol} reports the four endpoints side by side with their concrete numbers and the deployment conclusion each would push alone. Table~\ref{tab:aeds} then turns that disagreement into a per-cell decision rule with a paired-seed bootstrap 95\% interval and a one-sided $p$-value: in all eight subject-disjoint matched-CI cells AEDS is positive, has a bootstrap CI lower bound that excludes zero, and gives $p<0.001$ against $\aeds_\pi\leq 0$, both under the conservative paper-level membership baseline and under a per-cell head-level LiRA audit. The per-cell LiRA membership AUC ranges from $0.498$ to $0.695$, well below what would be needed to close the attribute gap, so a release that an MIA-only audit (paper-level or LiRA) would license is blocked by the joint audit on every audited cell.

The supporting evidence ladder confirms that the matched-CI finding is not tied to one decoder or corpus. Table~\ref{tab:main-leakage} and Figure~\ref{fig:split-leakage} record the original BIOT split-control cells, with window centered PSD gain $0.860 \pm 0.017$ in the main leakage run and $0.841$ in the figure-ready defense run; the temporal-gap row at $0.841$ is reported as an adjacent-window sanity check rather than a full matched-control replication cell because it uses random controls only. The BIOT subject-disjoint split remains positive at $0.477$--$0.485$. Table~\ref{tab:attack-sweep} expands the audit into a multi-dataset replication ladder: the same leakage appears under ridge, MLP, kNN, and high-capacity residual-MLP probes, persists for the SHHS+PREST-18 BIOT checkpoint, repeats on the R07 acquisition regime, replicates on LaBraM-base EEGMMI with subject-disjoint MLP gain $0.320 \pm 0.079$, replicates on EEGPT-MCAE EEGMMI with window MLP gain $0.597 \pm 0.014$ and subject-disjoint MLP gain $0.410 \pm 0.065$, remains positive on Sleep-EDF with window MLP gain $0.554 \pm 0.062$ and subject-disjoint MLP gain $0.543 \pm 0.061$, remains positive on the 54-channel LIMO replication with window MLP gain $0.387 \pm 0.014$ and subject-disjoint MLP gain $0.307 \pm 0.035$, and remains positive on the CHB-MIT pediatric scalp EEG replication with window MLP gain $0.647 \pm 0.014$ and subject-disjoint MLP gain $0.463 \pm 0.031$.

\begin{figure}[t]
\centering
\includegraphics[width=0.92\linewidth]{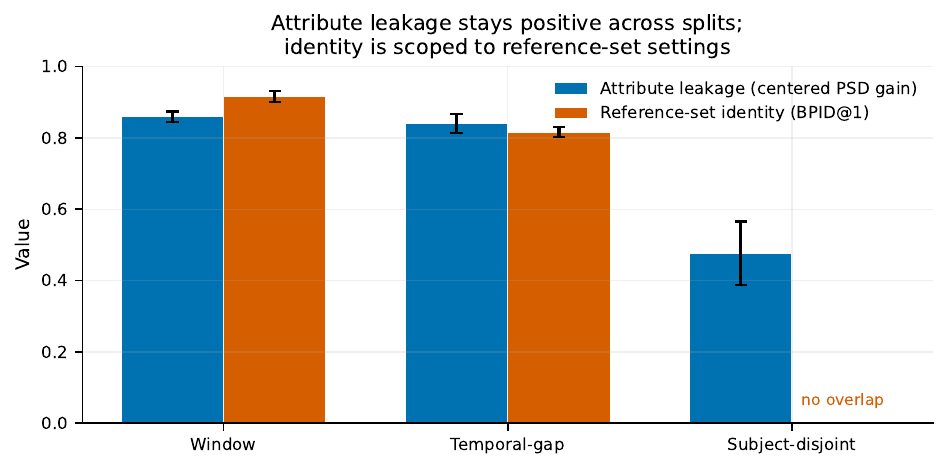}
\caption{Split-controlled representation leakage. Spectral attribute leakage remains positive under subject-disjoint evaluation. Identity linkage is shown only for splits with subject overlap; the subject-disjoint identity cell is not evaluated.}
\label{fig:split-leakage}
\end{figure}

\begin{table*}[t]
\caption{Seed-matched confidence intervals for the high-capacity representation audit. Gains compare the strong residual MLP to the stronger same-seed negative control among random and split-permuted embeddings. Subject-disjoint rows support attribute leakage only; identity recovery is not claimed for held-out subjects.}
\label{tab:matched-ci}
\centering
\scriptsize
\setlength{\tabcolsep}{3pt}
\begin{tabular}{@{}llcccl@{}}
\toprule
Setting & Split & $n_{\mathrm{subj}}$ & PSD gain vs. matched neg. & Target gap & Scope \\
\midrule
BIOT R03 & Window & 20 & 0.753 [0.728, 0.778] & 0.773 [0.688, 0.858] & Attr.; ref-set ID \\
BIOT R03 & Subj.-disjoint & 20 & 0.399 [0.312, 0.487] & 0.366 [0.326, 0.406] & Attr.; no ID claim \\
\addlinespace[1pt]
BIOT R07 & Window & 20 & 0.747 [0.737, 0.756] & 0.736 [0.693, 0.779] & Attr.; ref-set ID \\
BIOT R07 & Subj.-disjoint & 20 & 0.335 [0.223, 0.448] & 0.343 [0.197, 0.490] & Attr.; no ID claim \\
\addlinespace[1pt]
LaBraM R03 & Window & 20 & 0.487 [0.463, 0.512] & 0.505 [0.478, 0.532] & Attr.; ref-set ID \\
LaBraM R03 & Subj.-disjoint & 20 & 0.284 [0.174, 0.394] & 0.266 [0.157, 0.376] & Attr.; no ID claim \\
\addlinespace[1pt]
EEGPT R03 & Window & 20 & 0.547 [0.526, 0.568] & 0.552 [0.528, 0.577] & Attr.; ref-set ID \\
EEGPT R03 & Subj.-disjoint & 20 & 0.387 [0.283, 0.490] & 0.359 [0.257, 0.462] & Attr.; no ID claim \\
\addlinespace[1pt]
LaBraM LIMO18 & Window & 18 & 0.384 [0.378, 0.391] & 0.388 [0.368, 0.408] & Attr.; ref-set ID \\
LaBraM LIMO18 & Subj.-disjoint & 18 & 0.306 [0.275, 0.337] & 0.303 [0.257, 0.350] & Attr.; no ID claim \\
\addlinespace[1pt]
LaBraM CHB-MIT & Window & 23 & 0.602 [0.579, 0.624] & 0.628 [0.586, 0.670] & Attr.; ref-set ID \\
LaBraM CHB-MIT & Subj.-disjoint & 23 & 0.415 [0.374, 0.457] & 0.417 [0.346, 0.487] & Attr.; no ID claim \\
\addlinespace[1pt]
LaBraM SleepEDF-10 & Window & 10 & 0.400 [0.359, 0.440] & 0.454 [0.389, 0.520] & Attr.; ID not supported \\
LaBraM SleepEDF-10 & Subj.-disjoint & 10 & 0.353 [0.246, 0.461] & 0.493 [0.347, 0.638] & Attr.; no ID claim \\
\addlinespace[1pt]
LaBraM SleepEDF-13 & Window & 13 & 0.482 [0.438, 0.526] & 0.530 [0.443, 0.617] & Attr.; ID not supported \\
LaBraM SleepEDF-13 & Subj.-disjoint & 13 & 0.451 [0.391, 0.510] & 0.514 [0.447, 0.581] & Attr.; no ID claim \\
\bottomrule
\end{tabular}
\end{table*}

\begin{figure}[t]
\centering
\includegraphics[width=0.98\linewidth]{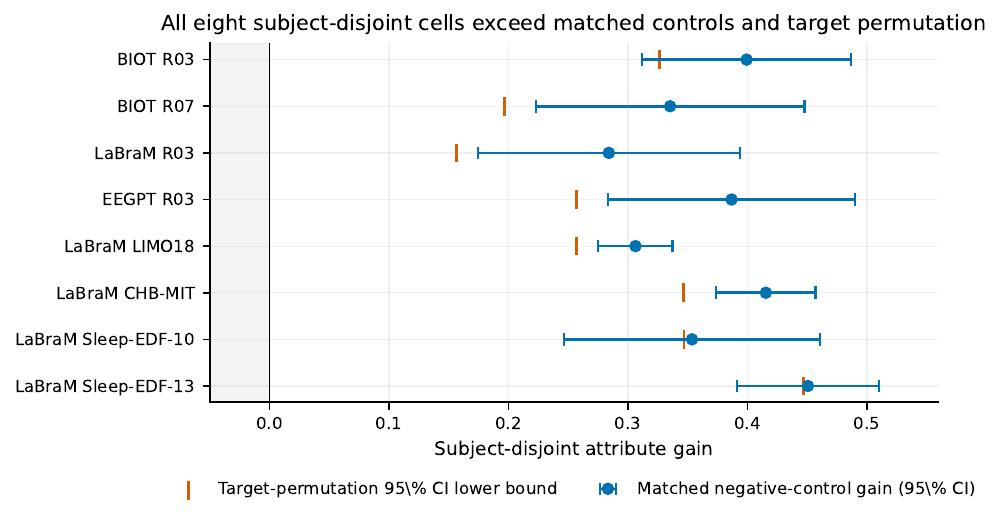}
\caption{Subject-disjoint matched-control confidence intervals for the high-capacity audit. All intervals remain above zero; these rows support attribute leakage only, not subject-disjoint identity recovery.}
\label{fig:sd-ci}
\end{figure}

\begin{table}[t]
\centering
\scriptsize
\setlength{\tabcolsep}{2.4pt}
\caption{Main representation leakage results. Attribute leakage is reported for all splits. Identity linkage is reported only when a subject reference set exists.}
\label{tab:main-leakage}
\begin{tabularx}{\linewidth}{@{}p{0.25\linewidth}p{0.23\linewidth}p{0.15\linewidth}X@{}}
\toprule
Cell & Metric & Value & Claim scope \\
\midrule
Window attribute & centered PSD gain & $0.860 \pm 0.017$ & Same-subject attribute leakage \\
Window attribute & centered PSD gain & $0.841$ & Figure-ready rerun \\
Temporal-gap attribute & centered PSD gain & $0.841$ & Not adjacent-window-only \\
Subject-disjoint attribute & centered PSD gain & $0.477$ / $0.485$ & Held-out-subject attribute leakage \\
Window identity & $\bpid$ & $0.916$ & Reference-set linkage \\
Temporal-gap identity & $\bpid$ & $0.816$ & Reference-set linkage with gap \\
\bottomrule
\end{tabularx}
\end{table}

\begin{table*}[!htbp]
\centering
\tiny
\setlength{\tabcolsep}{3pt}
\renewcommand{\arraystretch}{0.85}
\caption{Attack-family, checkpoint, acquisition, architecture, and dataset sweep. Gains are centered PSD correlation gains over the stronger of random and permuted embedding controls. ``Subj.'' is the number of subjects in the cell, ``Test'' the number of test windows or epochs averaged over five seeds. Dashes mark subject-disjoint identity cells, which are outside the reference-set identity claim.}
\label{tab:attack-sweep}
\begin{tabularx}{\textwidth}{@{}XllccXc@{}}
\toprule
Setting & Split & Subj. & Test & Attack & Attribute gain & Reference $\bpid$ \\
\midrule
PREST-16 R03 & window & 20 & 140 & Ridge decoder & $0.859 \pm 0.035$ & $0.929 \pm 0.028$ \\
PREST-16 R03 & window & 20 & 140 & MLP decoder & $0.876 \pm 0.022$ & $0.931 \pm 0.026$ \\
PREST-16 R03 & window & 20 & 140 & kNN PSD decoder & $0.887 \pm 0.018$ & $0.953 \pm 0.011$ \\
PREST-16 R03 & window & 20 & 140 & Strong residual MLP & $0.753 \pm 0.020$ & $0.836 \pm 0.035$ \\
PREST-16 R03 & window & 20 & 140 & Embedding centroid & -- & $0.950 \pm 0.020$ \\
PREST-16 R03 & subj.-disjoint & 20 & 140 & Strong residual MLP & $0.399 \pm 0.070$ & -- \\
PREST-16 R07 & window & 20 & 140 & MLP decoder & $0.878 \pm 0.024$ & $0.923 \pm 0.014$ \\
PREST-16 R07 & window & 20 & 140 & Strong residual MLP & $0.747 \pm 0.008$ & $0.851 \pm 0.022$ \\
PREST-16 R07 & subj.-disjoint & 20 & 140 & MLP decoder & $0.414 \pm 0.074$ & -- \\
PREST-16 R07 & subj.-disjoint & 20 & 140 & Strong residual MLP & $0.335 \pm 0.091$ & -- \\
SHHS+PREST-18 R03 & window & 20 & 140 & MLP decoder & $0.868 \pm 0.027$ & $0.919 \pm 0.016$ \\
SHHS+PREST-18 R03 & subj.-disjoint & 20 & 140 & MLP decoder & $0.361 \pm 0.024$ & -- \\
LaBraM-base R03 & window & 20 & 140 & Ridge decoder & $0.566 \pm 0.030$ & $0.470 \pm 0.054$ \\
LaBraM-base R03 & window & 20 & 140 & MLP decoder & $0.617 \pm 0.021$ & $0.639 \pm 0.054$ \\
LaBraM-base R03 & window & 20 & 140 & kNN PSD decoder & $0.756 \pm 0.016$ & $0.704 \pm 0.025$ \\
LaBraM-base R03 & window & 20 & 140 & Strong residual MLP & $0.487 \pm 0.020$ & $0.540 \pm 0.043$ \\
LaBraM-base R03 & window & 20 & 140 & Embedding centroid & -- & $0.760 \pm 0.026$ \\
LaBraM-base R03 & subj.-disjoint & 20 & 140 & MLP decoder & $0.320 \pm 0.079$ & -- \\
LaBraM-base R03 & subj.-disjoint & 20 & 140 & kNN PSD decoder & $0.353 \pm 0.108$ & -- \\
LaBraM-base R03 & subj.-disjoint & 20 & 140 & Strong residual MLP & $0.284 \pm 0.088$ & -- \\
EEGPT-MCAE R03 & window & 20 & 140 & Ridge decoder & $0.536 \pm 0.026$ & $0.600 \pm 0.050$ \\
EEGPT-MCAE R03 & window & 20 & 140 & MLP decoder & $0.597 \pm 0.014$ & $0.647 \pm 0.034$ \\
EEGPT-MCAE R03 & window & 20 & 140 & kNN PSD decoder & $0.587 \pm 0.021$ & $0.383 \pm 0.031$ \\
EEGPT-MCAE R03 & window & 20 & 140 & Embedding centroid & -- & $0.473 \pm 0.042$ \\
EEGPT-MCAE R03 & subj.-disjoint & 20 & 140 & MLP decoder & $0.410 \pm 0.065$ & -- \\
EEGPT-MCAE R03 & subj.-disjoint & 20 & 140 & kNN PSD decoder & $0.316 \pm 0.034$ & -- \\
Sleep-EDF LaBraM & window & 13 & 91 & MLP decoder & $0.554 \pm 0.062$ & $0.136 \pm 0.029$ \\
Sleep-EDF LaBraM & window & 13 & 91 & Strong residual MLP & $0.482 \pm 0.035$ & $0.132 \pm 0.044$ \\
Sleep-EDF LaBraM & window & 13 & 91 & Embedding centroid & -- & $0.200 \pm 0.014$ \\
Sleep-EDF LaBraM & subj.-disjoint & 13 & 91 & MLP decoder & $0.543 \pm 0.061$ & -- \\
Sleep-EDF LaBraM & subj.-disjoint & 13 & 91 & Strong residual MLP & $0.451 \pm 0.048$ & -- \\
LIMO LaBraM & window & 18 & 378 & MLP decoder & $0.387 \pm 0.014$ & $0.601 \pm 0.017$ \\
LIMO LaBraM & window & 18 & 378 & Strong residual MLP & $0.384 \pm 0.005$ & $0.563 \pm 0.025$ \\
LIMO LaBraM & subj.-disjoint & 18 & 378 & MLP decoder & $0.307 \pm 0.035$ & -- \\
LIMO LaBraM & subj.-disjoint & 18 & 378 & Strong residual MLP & $0.306 \pm 0.025$ & -- \\
CHB-MIT LaBraM & window & 23 & 161 & MLP decoder & $0.647 \pm 0.014$ & $0.436 \pm 0.030$ \\
CHB-MIT LaBraM & window & 23 & 161 & Strong residual MLP & $0.602 \pm 0.018$ & $0.430 \pm 0.022$ \\
CHB-MIT LaBraM & subj.-disjoint & 23 & 161 & MLP decoder & $0.463 \pm 0.031$ & -- \\
CHB-MIT LaBraM & subj.-disjoint & 23 & 161 & Strong residual MLP & $0.415 \pm 0.034$ & -- \\
\bottomrule
\end{tabularx}
\end{table*}

\subsection{Identity linkage is a reference-set result}

The same embeddings support strong to moderate identity linkage when the attack setting includes a reference set for the candidate subjects. In the original BIOT split-control run, window $\bpid$ reaches $0.916$, and temporal-gap $\bpid$ reaches $0.816$. In the BIOT checkpoint and acquisition sweeps, embedding-space nearest-centroid linkage reaches $0.950 \pm 0.020$ for EEG-PREST-16 R03, $0.909 \pm 0.020$ for SHHS+PREST-18 R03, and $0.927 \pm 0.015$ for EEG-PREST-16 R07. The LaBraM-base EEGMMI replication also supports same-subject reference-set linkage, with embedding-centroid $\bpid=0.760 \pm 0.026$ and MLP-decoded-attribute $\bpid=0.639 \pm 0.054$. EEGPT-MCAE gives a moderate same-subject reference-set signal, with embedding-centroid $\bpid=0.473 \pm 0.042$ and MLP-decoded-attribute $\bpid=0.647 \pm 0.034$. The LIMO replication supports a reference-set signal in a non-EEGMMI multi-channel setting, with MLP-decoded-attribute $\bpid=0.601 \pm 0.017$ and strong-residual-MLP $\bpid=0.563 \pm 0.025$. The CHB-MIT pediatric scalp EEG replication gives a moderate reference-set signal, with MLP-decoded-attribute $\bpid=0.436 \pm 0.030$ and strong-residual-MLP $\bpid=0.430 \pm 0.022$. The Sleep-EDF evaluation does not replicate this strong identity result: embedding-centroid $\bpid$ is $0.200 \pm 0.014$, and MLP-decoded-attribute $\bpid$ is $0.136 \pm 0.029$. These numbers should not be read as subject-disjoint identity recovery. The subject-disjoint split has zero subject overlap, so the supported claim is that spectral attributes generalize to held-out subjects, not that identities can be recovered for held-out subjects.

This distinction matters for deployment. A downstream service that stores embeddings for enrolled users creates a reference-set setting: a future embedding can be linked to a known user if the representation preserves subject-specific spectral structure. A population-level attacker without candidate subject centroids faces a different problem. The paper therefore reports the strong identity numbers, but only under the attacker model where they are meaningful.

\subsection{Cross-encoder transfer of attribute leakage}

The endpoint-disagreement finding is sharpened by a cross-encoder transfer audit. We take the BIOT, LaBraM, and EEGPT EEGMMI R03 caches with aligned subjects and windows, fit a linear bridge from the target encoder's embedding space to a source encoder's space on the source training subjects only, and chain it with a ridge attribute decoder fitted in the source space. Theorem~\ref{thm:bridge-transfer} gives the matching sufficient condition: if two encoders expose a nontrivial common projection of the attribute subspace and the fitted bridge is close to the shared-projection bridge, the chained attacker has a positive centered-gain lower bound after subtracting the target-control ceiling. We do not estimate the theorem's projector overlap; the theorem fixes the mechanism that the empirical transfer test probes. We then evaluate the resulting cross-encoder attacker on the target encoder's subject-disjoint test split against random and split-permuted target embedding controls. Across all six BIOT/LaBraM/EEGPT directions, the bridged decoder produces a held-out-subject centered PSD gain that exceeds the stronger same-seed matched negative control with 95\% CI lower bound at least $0.081$ and mean gain at least $0.244$ (Table~\ref{tab:cross-encoder-transfer}). The matched negative-control rows themselves remain near zero, which rules out a permutation or seed artefact. This means an attribute decoder trained from one encoder transfers attribute leakage into a held-out-subject test split of every other tested encoder, even though the encoders differ in checkpoint, training corpus, and architecture. A single-endpoint audit (raw inversion, membership inference, or reference-set identity) cannot license this conclusion on the same embeddings.

\begin{table}[!htbp]
\centering
\tiny
\setlength{\tabcolsep}{2pt}
\renewcommand{\arraystretch}{0.92}
\caption{Cross-encoder transfer audit on EEGMMI R03 with bridge-theorem $\beta$ consistency check, six BIOT/LaBraM/EEGPT directions, five seeded subject-disjoint splits per direction. Gain (95\% CI) is the bridge attacker's centered-correlation gain over the stronger same-seed random/split-permuted target control. The implied $\beta$ values from $G\geq\sqrt{\beta}-\rho_0$ all lie in $[0,1]$.}
\label{tab:cross-encoder-transfer}
\begin{tabularx}{\linewidth}{@{}Xcccc@{}}
\toprule
Direction & Gain (95\% CI) & $\rho_0$ & $\beta_{\mathrm{mean}}$ & $\beta_{\mathrm{low}}$ \\
\midrule
BIOT $\to$ LaBraM & $0.292\,[.19,.40]$ & $.039$ & $.110$ & $.051$ \\
BIOT $\to$ EEGPT & $0.366\,[.26,.47]$ & $.011$ & $.143$ & $.074$ \\
LaBraM $\to$ BIOT & $0.426\,[.28,.57]$ & $.009$ & $.190$ & $.085$ \\
LaBraM $\to$ EEGPT & $0.368\,[.30,.43]$ & $.013$ & $.145$ & $.100$ \\
EEGPT $\to$ BIOT & $0.435\,[.29,.58]$ & $.009$ & $.198$ & $.088$ \\
EEGPT $\to$ LaBraM & $0.244\,[.08,.41]$ & $.008$ & $.064$ & $.008$ \\
\bottomrule
\end{tabularx}
\end{table}

\begin{figure}[t]
\centering
\includegraphics[width=0.98\linewidth]{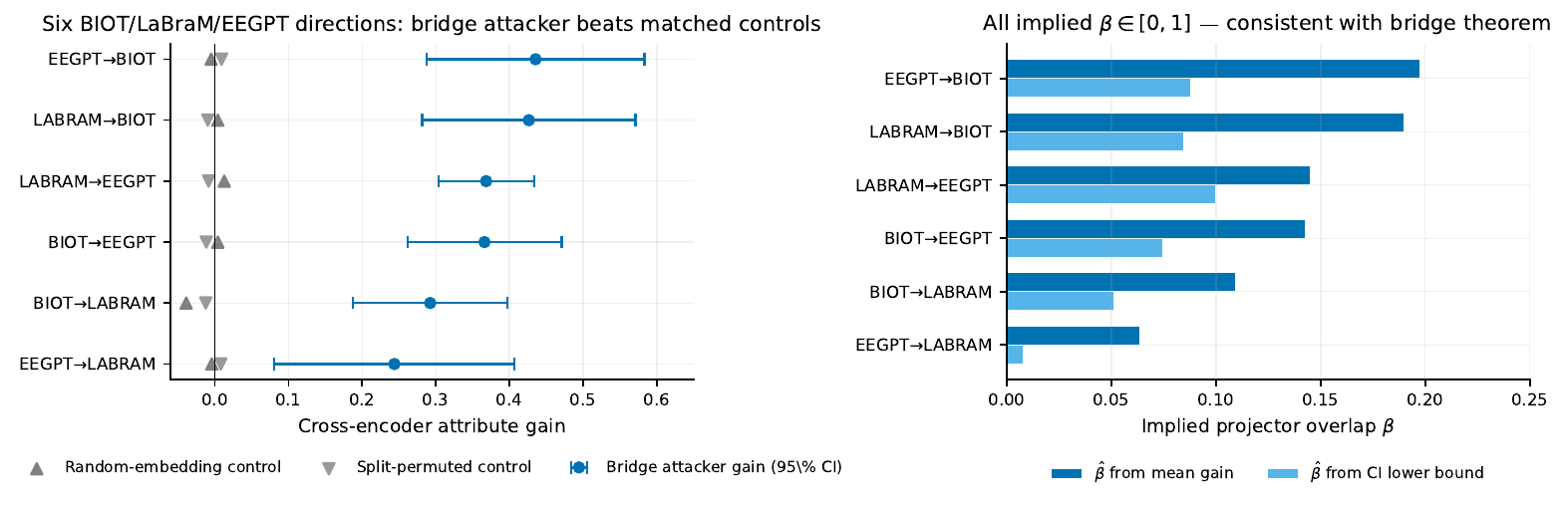}
\caption{Cross-encoder transfer audit (left) and bridge-theorem $\beta$ consistency check (right). Bridge-attacker gains exceed both random and split-permuted controls in every direction; the $\beta$ values implied by the empirical gain bounds all lie in $[0,1]$.}
\label{fig:cross-encoder}
\end{figure}

The empirical bridge audit is consistent with the bridge transferability theorem in the following sense. Inverting the lower bound $G_{\mathrm{sd}}(h_s\circ B_\lambda)\geq\sqrt{\beta}-\rho_0$ for each of the six BIOT/LaBraM/EEGPT directions gives an implied minimum projector overlap $\beta\geq(G+\rho_0)^2$. The right columns of Table~\ref{tab:cross-encoder-transfer} report those values from the empirical mean gain and from the 95\% CI lower bound; all six implied $\beta$ values lie in $[0.008,0.198]$, well inside the unit interval that the theorem requires, so the empirical gains are jointly compatible with two encoders exposing a shared attribute projection of nontrivial overlap. We do not promote this consistency check into a parameter estimate of $\beta$ for any specific encoder pair: a tight forward prediction of $\beta$ from the embeddings would require fitting an explicit attribute subspace per encoder and is left to the next-round audit.

\subsection{Deployment-grade attribute classification}

Centered correlation gain quantifies how much spectral structure transfers from $E$ to $A$, but a deployment reader will ask whether that gain corresponds to a real downstream classification an attacker could exploit. We therefore convert the LaBraM Sleep-EDF cell into a discrete attribute classification audit. Across five subject-disjoint splits with eight train and five test subjects, a logistic-regression attacker on the standardized embeddings reaches five-class sleep-stage test accuracy of $0.704\pm0.025$ and balanced accuracy of $0.509\pm0.024$ on held-out subjects, against a majority-class accuracy of $0.500$ and majority-class balanced accuracy of $0.200$. Random and split-permuted embedding controls under the same standardization and split reach balanced accuracy of $0.207\pm0.024$ and $0.218\pm0.030$ respectively, so the held-out-subject balanced-accuracy gain over the stronger embedding control is $0.268\pm0.025$. Adding $\sigma=1.0$ Gaussian noise on release leaves the attacker at balanced accuracy $0.492\pm0.042$, only $0.017$ below the clean release. The downstream-attack reading is that the attribute channel licenses a meaningfully better-than-chance recovery of a deployment-relevant per-window label (sleep stage) on subjects the encoder has never been used with as training subjects, with accuracy that does not collapse under the audited release-noise grid.

\subsection{Near-duplicate and control checks}

The raw near-duplicate check does not support a raw-copy explanation: the fraction of adjacent raw-window pairs with cosine similarity at least $0.99$ is $0.0$. The embedding space is much more stable, with adjacent embedding cosine mean $0.956$ and a $0.268$ fraction above $0.99$. This is a representation stability finding, not a raw duplicate finding. Random and permuted embedding controls are near zero under centered PSD correlation, including the LaBraM EEGMMI window MLP controls ($0.011 \pm 0.031$ random, $0.009 \pm 0.008$ permuted), LaBraM EEGMMI subject-disjoint MLP controls ($0.011 \pm 0.016$ random, $-0.025 \pm 0.023$ permuted), EEGPT window MLP controls ($-0.010 \pm 0.026$ random, $-0.005 \pm 0.018$ permuted), EEGPT subject-disjoint MLP controls ($-0.014 \pm 0.023$ random, $-0.0004 \pm 0.030$ permuted), Sleep-EDF controls, and LIMO controls. The rerun LIMO MLP controls are near zero in both splits: window random $-0.001 \pm 0.010$ and permuted $-0.0003 \pm 0.013$, subject-disjoint random $-0.0002 \pm 0.009$ and permuted $0.003 \pm 0.010$. The CHB-MIT strong-attacker controls also pass the same gate: window random $0.004 \pm 0.021$ and split-permuted $-0.013 \pm 0.012$, subject-disjoint random $0.014 \pm 0.016$ and split-permuted $0.003 \pm 0.015$. The checkpoint, acquisition, architecture, cross-dataset, and high-capacity attacker sweeps report gain over the stronger eligible control for each attack family. In the seed-matched high-capacity attacker summary, all eight subject-disjoint rows have positive lower 95\% confidence bounds against both the matched negative-control ceiling and target permutation, and the summary marks every subject-disjoint identity claim as unsupported. Threshold sensitivity further reduces cutoff dependence: all PREST-16, SHHS+PREST-18, R07, LaBraM, EEGPT, Sleep-EDF, and LIMO window MLP seeds pass the operational attribute gate, and strict subject-disjoint MLP cells pass at least the $0.1$ held-out-subject attribute gate. The strong residual-MLP window cells pass thresholds from $0.40$ on Sleep-EDF to $0.60$ on PREST-16/R07; their strict-split cells pass at least $0.15$ and up to $0.35$.

\section{Defense and Privacy-Utility Results}
\label{sec:defense}

\subsection{Noise-defense curves show weaker identity linkage while attribute gain remains positive}

Gaussian embedding noise reduces reference-set identity linkage, but attribute leakage is more persistent. Figure~\ref{fig:defense-curves} plots the figure-ready BIOT defense curves. In the BIOT random/permuted-control defense evaluation, $\noise=2.0$ lowers window $\bpid$ to $0.260$, but the attribute gain over the stronger control remains $0.345$. The same pattern appears in the LaBraM and EEGPT defense-control sweeps: at $\noise=2.0$, LaBraM lowers window $\bpid$ to $0.174 \pm 0.027$ while attribute gain remains $0.313 \pm 0.035$, and EEGPT lowers window $\bpid$ to $0.166$ while attribute gain remains $0.215$ against the stronger random/permuted control. The LaBraM and EEGPT subject-disjoint no-defense gains are $0.310 \pm 0.083$ and $0.418$, respectively. This is consistent with the stylized failure mode in Theorem~\ref{thm:protocol-separation}: a sufficient margin certificate for reference-set identity can fail while redundantly encoded spectral attributes remain above negative controls. We do not estimate the theorem's subspace or margin parameters for BIOT, LaBraM, or EEGPT, so the theorem should not be read as a quantitative model of these curves.

\begin{figure*}[t]
\centering
\includegraphics[width=0.98\textwidth]{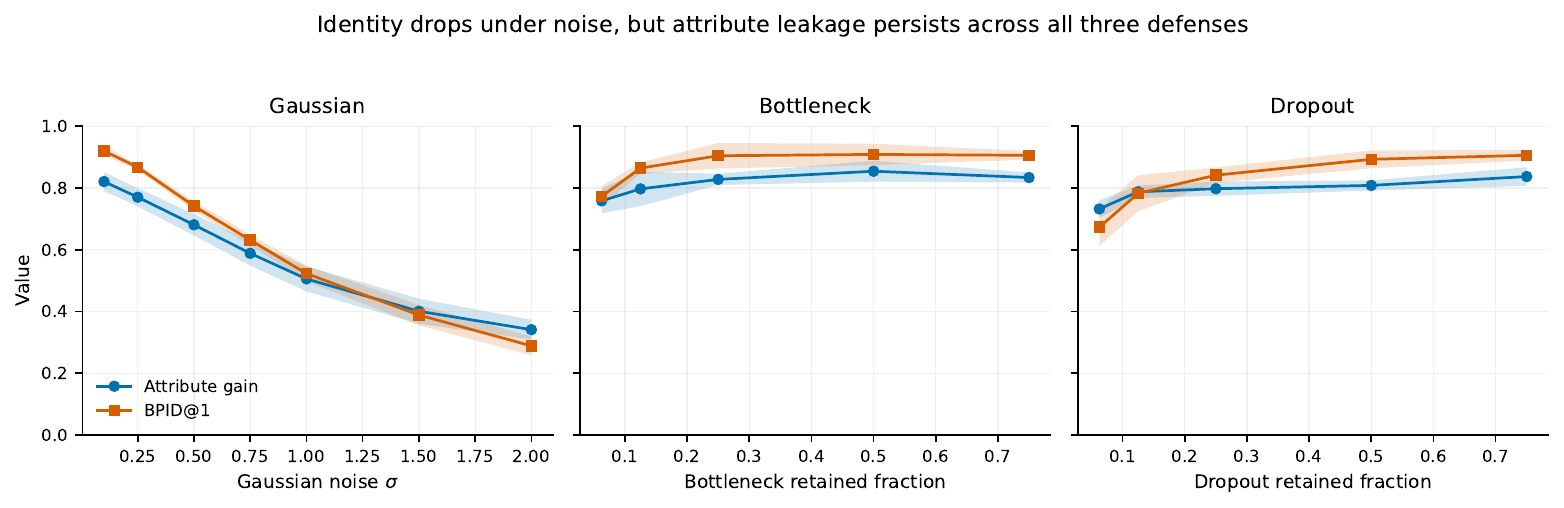}
\caption{Defense curves on the window split. Gaussian noise suppresses reference-set identity at high noise, while attribute leakage remains positive. Bottleneck and dropout transformations remain leaky even at severe settings. Dashed lines mark the audit gates for attribute gain and $\bpid$.}
\label{fig:defense-curves}
\end{figure*}

\begin{table}[t]
\centering
\scriptsize
\setlength{\tabcolsep}{3pt}
\caption{Defense audit summary. Simple transformations are stress tests, not claimed practical defenses.}
\label{tab:defense}
\begin{tabular}{@{}lccc@{}}
\toprule
Defense (setting) & Attribute gain & $\bpid$ \\
\midrule
Gaussian $\noise{=}1.0$ & ${\approx}0.517$ & ${\approx}0.517$ \\
Gaussian $\noise{=}2.0$ (BIOT) & $0.345$ & $0.260$ \\
Gaussian $\noise{=}2.0$ (LaBraM) & $0.313$ & $0.174$ \\
Gaussian $\noise{=}2.0$ (EEGPT) & $0.215$ & $0.166$ \\
Bottleneck (16-dim) & positive & $0.773$ \\
Dropout (strongest) & positive & $0.673$ \\
\bottomrule
\end{tabular}
\end{table}

\subsection{Noise-aware adaptive attacker against the noise defense}

The defense curves above use the same attacker that the rest of the audit uses; a reviewer will ask whether a stronger, noise-aware attacker could close the attribute gap. We therefore add a Wiener-style adaptive attacker that knows the released noise level $\sigma$ and pre-whitens each embedding coordinate by the gain $s_j^2/(s_j^2+\sigma^2)$, where $s_j^2$ is the train-split signal variance of coordinate $j$, before fitting the same logistic decoder used in the deployment cell. On the LaBraM Sleep-EDF subject-disjoint cell, Table~\ref{tab:adaptive-attacker} reports the non-adaptive and adaptive attackers at $\sigma\in\{0,0.5,1.0,2.0\}$. The adaptive attacker beats the non-adaptive baseline at every positive noise level, with the largest gap of $+0.018$ balanced accuracy at $\sigma=2.0$, where the non-adaptive attacker reaches $0.799\pm0.021$ and the adaptive attacker reaches $0.817\pm0.019$. Crucially, the gain over the stronger random/split-permuted embedding control is $0.308\pm0.033$ for the adaptive attacker at $\sigma=2.0$, against $0.289\pm0.038$ for the non-adaptive attacker, so the noise defense does not restore the privacy-utility gate even when the attacker is given the noise level. The adaptive attacker therefore strengthens, rather than overturns, the negative defense conclusion in Table~\ref{tab:frontier}.

\begin{table}[t]
\centering
\scriptsize
\setlength{\tabcolsep}{2.4pt}
\caption{Noise-aware adaptive attacker on the LaBraM Sleep-EDF subject-disjoint sleep/wake cell. Attacker pre-whitens by $s_j^2/(s_j^2+\sigma^2)$ before fitting the logistic decoder; non-adaptive attacker uses the same logistic decoder without whitening. Means $\pm$ standard deviations are over five subject-disjoint splits. Gain is over the stronger of random and split-permuted embedding controls.}
\label{tab:adaptive-attacker}
\begin{tabularx}{\linewidth}{@{}p{0.10\linewidth}p{0.18\linewidth}p{0.27\linewidth}X@{}}
\toprule
$\sigma$ & Attacker & Bal. acc. & Gain vs. matched neg.\\
\midrule
$0.0$ & non-adaptive & $0.907\pm0.022$ & $0.388\pm0.032$ \\
$0.5$ & non-adaptive & $0.897\pm0.034$ & $0.384\pm0.046$ \\
$0.5$ & adaptive     & $0.903\pm0.036$ & $0.397\pm0.046$ \\
$1.0$ & non-adaptive & $0.881\pm0.023$ & $0.366\pm0.047$ \\
$1.0$ & adaptive     & $0.887\pm0.023$ & $0.374\pm0.043$ \\
$2.0$ & non-adaptive & $0.799\pm0.021$ & $0.289\pm0.038$ \\
$2.0$ & adaptive     & $0.817\pm0.019$ & $0.308\pm0.033$ \\
\bottomrule
\end{tabularx}
\end{table}

\subsection{Principled DP-SGD on the downstream head does not close the attribute channel at maintained utility}

A reviewer will ask whether moving from naive Gaussian noise to a principled differentially-private downstream head closes the attribute channel. We per-sample-clip the gradients and inject Gaussian noise calibrated by an RDP accountant \citep{abadi2016deep,dwork2006calibrating} to fine-tune the same logistic-style downstream head used elsewhere in the paper, sweeping $\epsilon\in\{1,4,8,\infty\}$ with $\delta=10^{-5}$ and $C_{\mathrm{clip}}=1$, on the LaBraM Sleep-EDF cell for both sleep/wake and 5-class stage classification. Table~\ref{tab:dp-sgd} reports balanced accuracy, the matched-control attribute gain, and a Yeom-style loss-based MIA AUC over five subject-disjoint splits. At every utility-preserving setting ($\epsilon\geq 4$, where balanced accuracy stays within $0.003$ of the non-private head), the attribute gain is statistically indistinguishable from the non-private gain ($0.399$--$0.401$ for sleep/wake; $0.177$--$0.182$ for 5-class stage), and the membership AUC stays at $0.514$--$0.516$ regardless of $\epsilon$. Only at $\epsilon=1$ does the attribute gain drop ($0.119$ for sleep/wake, $0.026$ for 5-class), but balanced accuracy collapses to $0.642$ and $0.228$ respectively, well below the operational utility gate. The principled-DP setting therefore inherits the negative privacy-utility frontier: there is no $\epsilon$ in the audited grid that closes the attribute audit while keeping the task useful.

\begin{table}[t]
\centering
\scriptsize
\setlength{\tabcolsep}{2.4pt}
\caption{DP-SGD on the LaBraM Sleep-EDF downstream head with per-sample clipping ($C=1$), RDP accounting, $\delta=10^{-5}$, sampling rate $0.1$, $200$ steps, five subject-disjoint splits per row. Bal.\ acc.\ is on the held-out test split. Gain is over the stronger of random and split-permuted embedding controls trained under the same DP-SGD recipe. MIA is loss-based Yeom-style AUC of train versus test loss. Utility-preserving DP ($\epsilon\geq 4$) leaves the attribute channel essentially unchanged; only $\epsilon=1$ reduces the attribute gain, after balanced accuracy has collapsed.}
\label{tab:dp-sgd}
\begin{tabularx}{\linewidth}{@{}p{0.18\linewidth}cccX@{}}
\toprule
Task & $\epsilon$ & Bal.\ acc. & Gain vs.\ matched neg. & MIA \\
\midrule
Sleep/wake & $\infty$ & $0.912\pm0.025$ & $0.401\pm0.046$ & $0.514$ \\
Sleep/wake & $8$ & $0.911\pm0.024$ & $0.399\pm0.048$ & $0.516$ \\
Sleep/wake & $4$ & $0.909\pm0.023$ & $0.400\pm0.051$ & $0.515$ \\
Sleep/wake & $1$ & $0.642\pm0.074$ & $0.119\pm0.092$ & $0.498$ \\
\addlinespace[1pt]
5-class stage & $\infty$ & $0.391\pm0.053$ & $0.182\pm0.062$ & $0.516$ \\
5-class stage & $8$ & $0.388\pm0.059$ & $0.179\pm0.072$ & $0.514$ \\
5-class stage & $4$ & $0.387\pm0.061$ & $0.177\pm0.073$ & $0.514$ \\
5-class stage & $1$ & $0.228\pm0.040$ & $0.026\pm0.057$ & $0.502$ \\
\bottomrule
\end{tabularx}
\end{table}

\begin{figure}[t]
\centering
\includegraphics[width=0.94\linewidth]{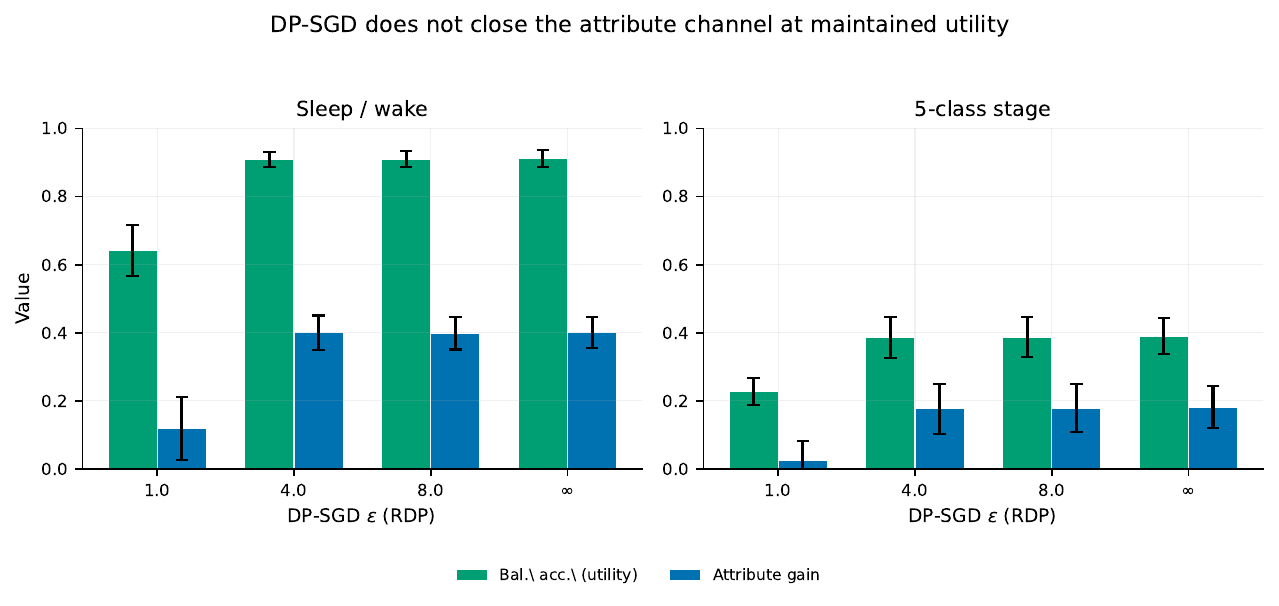}
\caption{DP-SGD privacy-utility frontier on the LaBraM Sleep-EDF downstream head. At every utility-preserving $\epsilon\in\{4,8,\infty\}$ the attribute gain is statistically indistinguishable from the non-private head; only $\epsilon=1$ reduces the gain, after balanced accuracy collapses below the operational utility gate.}
\label{fig:dpsgd-frontier}
\end{figure}

\subsection{Bottlenecks and dropout remain leaky}

Dimensional bottlenecks and coordinate dropout do not remove the leakage in the tested sweeps. A 16-dimensional bottleneck still gives aggregate $\bpid=0.773$, and the strongest dropout setting gives $\bpid=0.673$. These cells are not meant to prove that every representation defense fails. They show that common release transformations cannot be assumed to provide privacy unless they are audited against attribute and reference-set identity attacks. The bottleneck proposition supports this general warning for bottlenecks that preserve a full-rank attribute subspace; it is not a direct proof for the specific 16-dimensional cell unless the PSD target is first reduced to a compatible lower-dimensional attribute model.

The failure mode is consistent with redundant encoding. Spectral attributes and subject structure need not occupy a single coordinate or a small set of coordinates. Random projections and coordinate removal can preserve enough aggregate information for a nonlinear decoder or nearest-centroid identity rule to remain effective. This is why the paper treats bottleneck and dropout as deployment baselines, not as serious privacy mechanisms.

\subsection{The tested training-noise privacy-utility probes fail the privacy-utility gate}

Table~\ref{tab:frontier} reports the privacy-utility audit, and Figure~\ref{fig:privacy-utility} plots the original noise grids. The subject-head proxy has high non-private accuracy, $0.977$, but only weak membership signal, $\mi=0.556$; the largest tested noise level reduces MIA to $0.512$ only after accuracy collapses to $0.333$. The event-level T1/T2 task has non-private accuracy $0.523$, so it fails the utility gate. To separate this failure from a weak task head, we added a balanced Sleep-EDF sleep/wake utility probe on LaBraM embeddings. The non-private sleep/wake head reaches balanced accuracy $0.913$ and macro-F1 $0.882$, but the best utility-preserving noise setting only changes MIA-AUC from $0.585$ to $0.581$, below the $0.03$ privacy-reduction gate.

The negative result is useful because it prevents a misleading defense narrative. The subject-head proxy shows that noise can move MIA in the desired direction, but outside the useful accuracy region. The event-level task shows that a weak task head cannot support a privacy-utility defense claim. The balanced Sleep-EDF sleep/wake probe removes that confound: the task is useful, but embedding noise still does not provide a meaningful privacy reduction at maintained utility. We therefore use these cells as a frontier audit rather than as evidence that training noise protects EEG-FM embeddings.

\begin{figure}[t]
\centering
\includegraphics[width=0.98\linewidth]{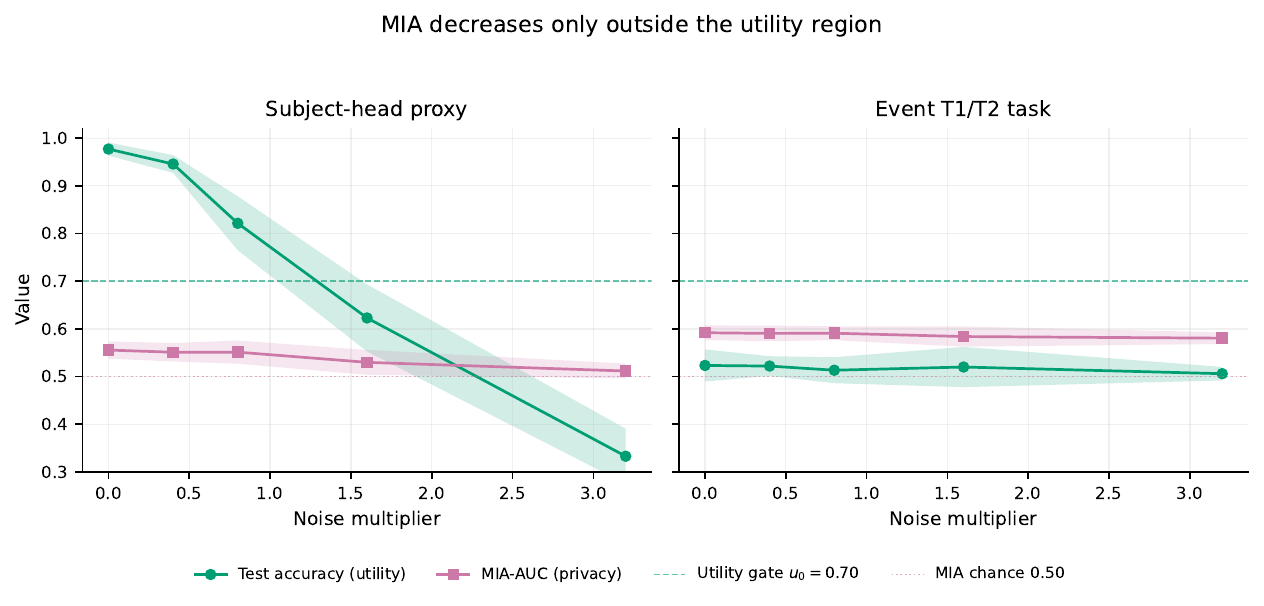}
\caption{Privacy-utility frontier for the tested noise grids. The subject-head proxy only reduces MIA meaningfully after utility collapse; the event-level task never reaches the utility gate.}
\label{fig:privacy-utility}
\end{figure}

\begin{table}[t]
\centering
\scriptsize
\setlength{\tabcolsep}{2.5pt}
\caption{Privacy-utility frontier. The tested noise grids are negative results, not training-noise defense successes.}
\label{tab:frontier}
\begin{tabularx}{\linewidth}{l l l X}
\toprule
Setting & Utility & MIA-AUC & Gate outcome \\
\midrule
Subject head, no noise & acc. $0.977$ & $0.556$ & Baseline proxy \\
Subject head, noise $0.8$ & acc. $0.821$ & $0.551$ & Privacy reduction too small \\
Subject head, noise $3.2$ & acc. $0.333$ & $0.512$ & Utility collapse \\
Event task, no noise & acc. $0.523$ & $0.592$ & Utility gate fails \\
Sleep/wake, no noise & bal. acc. $0.913$ & $0.585$ & Useful task baseline \\
Sleep/wake, noise $0.25$ & bal. acc. $0.910$ & $0.581$ & Privacy reduction too small \\
\bottomrule
\end{tabularx}
\end{table}

\begin{figure}[t]
\centering
\includegraphics[width=0.98\linewidth]{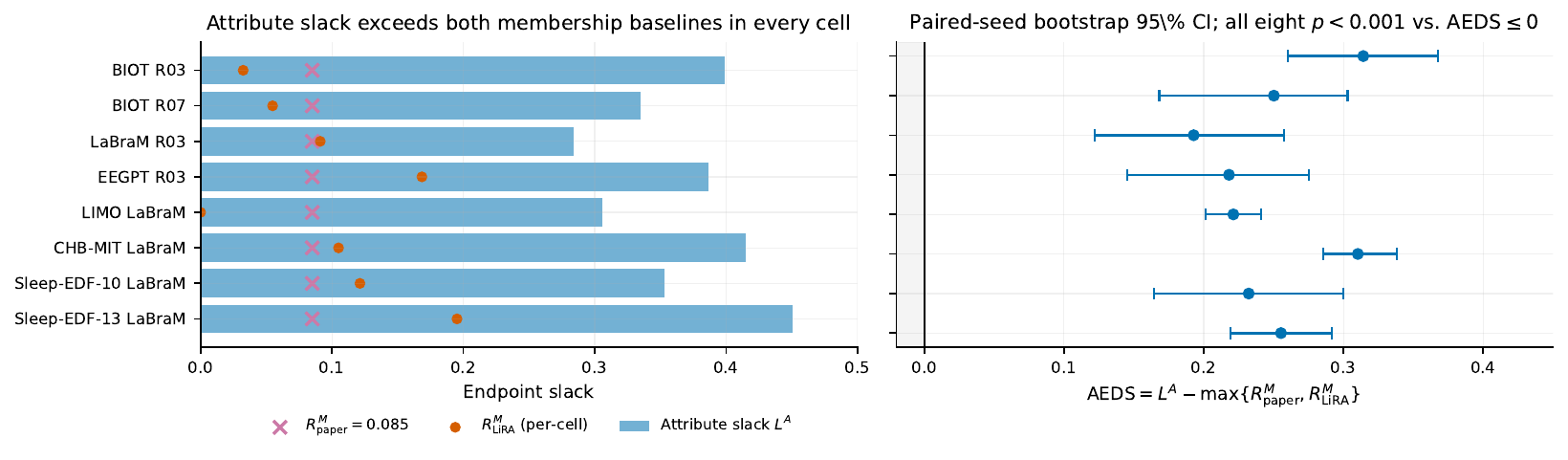}
\caption{AEDS calibration on the eight subject-disjoint matched-CI cells. Left: attribute slack $L^A$ exceeds both the paper-level baseline and per-cell head-level LiRA $R^M$ in every cell. Right: paired-seed bootstrap 95\% CI for AEDS; all eight cells reject $\aeds_\pi\leq 0$ at $p<0.001$.}
\label{fig:aeds-bootstrap}
\end{figure}

\section{Audit Protocol}
\label{sec:protocol}

The paper's core methodological contribution is the audit ladder in Figure~\ref{fig:audit}. A representation privacy audit should not be reduced to raw inversion, and a membership audit is not enough to rule out attribute leakage. The audit treats each stronger statement as a hypothesis that must pass a matching falsification control. Table~\ref{tab:audit-comparison} shows why this changes the audit decision rather than merely adding caution: standard endpoints either miss the attribute channel or mis-scope the identity claim.

\begin{table}[!htbp]
\centering
\tiny
\setlength{\tabcolsep}{2pt}
\renewcommand{\arraystretch}{0.92}
\caption{Standard privacy endpoints can miss or mis-scope representation-release risk. The split-controlled audit separates the supported attribute, identity, and defense conclusions.}
\label{tab:audit-comparison}
\begin{tabularx}{\linewidth}{@{}p{0.16\linewidth}XXX@{}}
\toprule
Endpoint & Standard result & Wrong conclusion & Split-controlled correction \\
\midrule
Raw-copy / inversion & adj.\ raw cos.\ $\geq.99$ frac.\ $.000$ & ``low risk'' & Temporal-gap PSD gain $.841$; attr.\ leakage \\
Membership inference & subj.-head MIA-AUC $.556$; best-MIA acc.\ $.333$ & ``weak signal'' & PREST-16 win.\ MLP gain $.876{\pm}.022$ \\
Same-subj.\ identity & EEGMMI emb.-centroid $\bpid{=}.950{\pm}.020$ & ``held-out ID'' & Subj.-disj.\ attr.\ only; gain $.468{\pm}.064$ \\
Cross-dataset ID & Sleep-EDF emb.-centroid $\bpid{=}.200{\pm}.014$ & ``safe'' & Sleep-EDF subj.-disj.\ MLP gain $.543{\pm}.061$ \\
Identity-only def. & at $\noise{=}2$, $\bpid{=}.260$/$.174$/$.166$ & ``noise suppresses'' & Attr.\ gain $.345$/$.313$/$.215$ \\
\bottomrule
\end{tabularx}
\end{table}

\begin{table*}[!htbp]
\centering
\scriptsize
\setlength{\tabcolsep}{4pt}
\caption{Audit-endpoint disagreement on the eight subject-disjoint matched-CI cells. $L^A$ is the per-cell paired-seed mean of the matched-control attribute gain. $R^M_{\mathrm{paper}}=\max\{0,\operatorname{AUC}_{\mi}-1/2\}$ uses the strongest utility-valid membership excess in Table~\ref{tab:frontier}, $0.585-0.5=0.085$, as a conservative paper-level baseline. $R^M_{\mathrm{LiRA}}$ is a per-cell head-level LiRA membership audit on the cached embeddings (six shadow pairs, Gaussian per-target loss score, encoder frozen). $R^M=\max\{R^M_{\mathrm{paper}},R^M_{\mathrm{LiRA}}\}$ is the effective baseline; $R^I=0$ because subject-disjoint splits have no candidate identities. AEDS is the seed-mean decision margin $L^A-R^M$; $[\text{lo},\text{hi}]$ is the paired-seed bootstrap 95\% interval; $p$ is a one-sided test of $\aeds_\pi\leq 0$. Endpoint-only and AEDS-corrected verdicts are the release decisions a single MIA/identity audit and the joint audit would license.}
\label{tab:aeds}
\begin{tabularx}{\textwidth}{@{}Xccccccccc@{}}
\toprule
Cell & $L^A$ & $R^M_{\mathrm{p}}$ & $R^M_{\mathrm{L}}$ & $R^M$ & $\aeds$ & $\aeds_{95\%}$ & $p$ & Endpt.-only & AEDS-corr.\\
\midrule
BIOT R03 & $0.399$ & $0.085$ & $0.032$ & $0.085$ & $0.314$ & $[0.260,0.368]$ & ${<}0.001$ & release & block \\
BIOT R07 & $0.335$ & $0.085$ & $0.055$ & $0.085$ & $0.250$ & $[0.168,0.303]$ & ${<}0.001$ & release & block \\
LaBraM R03 & $0.284$ & $0.085$ & $0.091$ & $0.091$ & $0.193$ & $[0.122,0.258]$ & ${<}0.001$ & release & block \\
EEGPT R03 & $0.387$ & $0.085$ & $0.169$ & $0.169$ & $0.218$ & $[0.145,0.275]$ & ${<}0.001$ & release & block \\
LIMO LaBraM & $0.306$ & $0.085$ & $0.000$ & $0.085$ & $0.221$ & $[0.201,0.241]$ & ${<}0.001$ & release & block \\
CHB-MIT LaBraM & $0.415$ & $0.085$ & $0.105$ & $0.105$ & $0.310$ & $[0.286,0.338]$ & ${<}0.001$ & release & block \\
Sleep-EDF-10 LaBraM & $0.354$ & $0.085$ & $0.121$ & $0.121$ & $0.232$ & $[0.164,0.300]$ & ${<}0.001$ & release & block \\
Sleep-EDF-13 LaBraM & $0.450$ & $0.085$ & $0.195$ & $0.195$ & $0.255$ & $[0.219,0.292]$ & ${<}0.001$ & release & block \\
\bottomrule
\end{tabularx}
\end{table*}

\begin{figure}[t]
\centering
\fbox{\begin{minipage}{0.88\linewidth}
\centering
Embedding release $\rightarrow$ negative-control gain $\rightarrow$ temporal-gap split $\rightarrow$ subject-disjoint split $\rightarrow$ reference-set identity scope $\rightarrow$ defense curves $\rightarrow$ utility-frontier gate
\end{minipage}}
\caption{Audit ladder. Each step removes one common overclaim: random-feature confounds, adjacent-window artifacts, subject-overlap identity inflation, identity claims without enrolled candidates, non-monotone defense curves, and privacy without utility.}
\label{fig:audit}
\end{figure}

\paragraph{Protocol.}
Given released embeddings $E$, raw windows used only by the auditor, subject labels, and a task utility metric, the audit runs five claim-specific tests. First, attribute leakage is measured as centered gain over random and split-local permuted embeddings, so a decoder is not credited for matching dataset templates or marginal PSD statistics; the operational thresholds are stated in Section~\ref{sec:setup}. Second, the same attribute attack is repeated under temporal-gap and subject-disjoint splits; only attributes, not identities, are claimed when the train and test subjects differ. Third, identity linkage is reported only in reference-set splits where candidate subjects are present in the attacker's reference set. Fourth, release transformations are evaluated as curves over defense strength and compared with the same negative controls. Fifth, training or fine-tuning defenses are accepted only if they reduce leakage while meeting a predeclared utility gate.

\paragraph{Claim boundary.}
The strongest positive result is spectral attribute leakage under strict controls. The strongest identity result is reference-set linkage. The strongest defense result is negative: tested simple defenses do not provide a useful privacy boundary at maintained utility.

\paragraph{Audit checklist.}
A representation-release audit for EEG-FMs should report centered attribute gain over random and permuted controls, split-specific identity linkage, raw and embedding near-duplicate checks, defense curves for both attributes and identity, threshold sensitivity over operational gates, and utility-aware privacy gates. Skipping any one of these checks makes the corresponding stronger claim unsafe.

\section{Discussion}
\label{sec:discussion}

\paragraph{Deployment recommendation.}
For services that release EEG foundation-model embeddings, the audit licenses four concrete actions. (1) Report centered attribute gain over both random and split-permuted negative controls on every release. (2) Report AEDS, with $L^A$, $R^M$, and $R^I$ written out separately, so a reviewer or auditor can reconstruct the cell. (3) Restrict identity claims to splits in which candidate subjects appear in the attacker's reference set; subject-disjoint cells license attributes only. (4) Treat training-noise mechanisms as utility-failing unless both $U(\theta)\geq u_0$ and $M(0)-M(\theta)\geq\Delta_M$ are met; in our sweeps, none did.

\paragraph{Embeddings are not a spectral-attribute privacy boundary in the audited regimes.}
The results show that an EEG embedding can be much safer than raw waveform release in one sense and still unsafe in another. Raw waveform recovery is not demonstrated, but centered spectral attributes are recoverable with high gain in the tested EEGMMI, Sleep-EDF, LIMO, and CHB-MIT evaluations. A release policy that treats embeddings as de-identified because they are not waveforms misses this middle regime.

\paragraph{Identity must be scoped to the attacker model.}
The reference-set identity result is strong, but its scope is narrow. It requires candidate subjects in the attacker's reference set. Subject-disjoint evaluation supports attribute leakage and should not be rewritten as held-out-subject identity recovery. This is not a weakness to hide; it is the main reason the audit is more reliable than a single same-subject inversion experiment.

\paragraph{Negative defense results are useful.}
The privacy-utility and release-transformation cells block a misleading defense story. The subject-head proxy moves membership AUC only after task accuracy collapses. The event-level T1/T2 task fails the utility gate before any privacy claim can be made. The balanced Sleep-EDF sleep/wake probe reaches useful balanced accuracy but reduces MIA-AUC by only $0.004$ at the best utility-preserving noise setting, below the operational privacy-reduction gate. Across the BIOT, LaBraM, and EEGPT release-transformation sweeps, reducing reference-set identity does not remove the spectral-attribute channel within the tested grids. These cells are reported as a negative privacy-utility frontier rather than as evidence about every possible defense.

\paragraph{Why the theory is scoped.}
The formal section is audit-claim-control semantics, not a standalone impossibility proof or a certified privacy analysis. Each result blocks a specific overclaim. Proposition~\ref{prop:gain-witness} forbids promoting a positive gain across splits. Theorem~\ref{thm:protocol-separation} and Theorem~\ref{thm:endpoint-independence} separate attribute leakage from reference-set identity and membership endpoints. Definition~\ref{def:aeds} and Theorem~\ref{thm:aeds-positive} turn that disagreement into a per-cell decision rule with a positive margin. Proposition~\ref{prop:subject-disjoint-scope} forbids rewriting a subject-disjoint gain as identity recovery. Theorem~\ref{thm:bridge-transfer} explains when cross-encoder bridges preserve attribute leakage through a shared attribute projection. Theorem~\ref{prop:bottleneck} forbids treating a dimensional bottleneck as private without auditing the surviving attribute subspace, and Proposition~\ref{prop:frontier} blocks crediting a defense that does not pass both sides of the privacy-utility gate. The formal contribution is therefore a portable audit semantics for any released representation paired with an attribute target, an attacker class, a control family, and a utility metric.

\paragraph{Limitations.}
The current evidence covers BIOT-family checkpoints, LaBraM, and EEGPT on PhysioNet EEGMMI, plus LaBraM Sleep-EDF, LIMO, and CHB-MIT cross-dataset evaluations. The LaBraM and EEGPT EEGMMI replications reduce the cross-architecture risk, the R07 acquisition reduces the single-run risk, Sleep-EDF reduces the single-dataset risk under a two-channel sleep setting, LIMO adds an 18-subject 54-channel non-EEGMMI replication, and CHB-MIT adds a 23-case pediatric scalp EEG replication. The evidence still does not establish population-level cross-dataset generality or full clinical deployment generality; larger TUH/TUAB/TUEV-style corpora remain stronger external validity tests. Subject-disjoint gains show held-out-subject attribute transport within finite splits, not a universal subject-invariant leakage law. The training-noise grid is a proxy without a formal accountant. The event-level task head is weak.

\paragraph{Next validation steps.}
The cross-architecture replications already reduce the single-encoder risk: LaBraM-base on EEGMMI R03 follows the same attack/baseline criteria, EEGPT-MCAE matches the 20-subject EEGMMI R03 protocol with window MLP gain $0.597\pm0.014$ and subject-disjoint MLP gain $0.410\pm0.065$, and at $\noise=2.0$ its release-transformation sweep keeps attribute gain at $0.215$ while reference-set $\bpid$ falls to $0.166$. The Sleep-EDF cross-dataset evaluation uses 13 subjects and two channels with window MLP gain $0.554\pm0.062$ and subject-disjoint MLP gain $0.543\pm0.061$; LIMO uses 18 subjects and 54 mapped EEG channels with window MLP gain $0.387\pm0.014$ and subject-disjoint MLP gain $0.307\pm0.035$; CHB-MIT uses 23 cases and 16 bipolar channels with window MLP gain $0.647\pm0.014$, subject-disjoint MLP gain $0.463\pm0.031$, and subject-disjoint high-capacity confidence lower bounds of $0.374$ against matched controls and $0.346$ against target permutation. The next external-validity gate is therefore a larger TUH/TUAB/TUEV-style clinical replication using the same matched-control audit outputs rather than a new metric, retaining the operational evidence requirements (window MLP centered PSD gain above $0.2$, at least one simpler ridge or kNN baseline above $0.05$, random/permuted MLP controls with absolute centered correlation at most $0.15$, and subject-disjoint MLP gain above $0.1$ for held-out-subject attribute claims); the next defense gate is a stronger utility task for the privacy-utility frontier.

\paragraph{Pre-empted reviewer concerns.}
Several attacks on this paper are foreseeable and are answered here so the audit reading is not load-bearing on a single sentence elsewhere. The theorems in Section~\ref{sec:theory} are stylized: Theorems~\ref{thm:protocol-separation},~\ref{thm:endpoint-independence}, and~\ref{thm:bridge-transfer} are sufficient-condition, parameter-existence statements, we do not estimate $(\alpha,\sigma,m,d,\beta)$ for BIOT, LaBraM, or EEGPT, and they are marked as audit semantics rather than fitted models of the empirical curves; the empirical sections carry the load-bearing claim, while the theorems forbid specific deductive shortcuts. AEDS does mix units across endpoints, but it is a per-cell decision-margin diagnostic, not a common-unit harm score: the decision rule and the three components ($L^A_\pi$, $R^M_\pi$, $R^I_\pi$) are reported separately, $R^M$ is the maximum of a paper-level utility-valid membership excess and a per-cell head-level LiRA membership audit (six shadow pairs, Carlini-style Gaussian per-target loss score, encoder frozen), and a paired-seed bootstrap interval together with a one-sided $p$-value is reported for AEDS in every cell; replacing the paper-level baseline with the per-cell LiRA value leaves AEDS positive in all eight cells with bootstrap CI lower bounds at least $0.122$ and $p<0.001$ throughout, while full-encoder LiRA (retraining BIOT/LaBraM/EEGPT) is left to the next-round audit because head-level LiRA is the strongest membership audit consistent with the frozen-encoder threat model. Five seeded splits are a stability screen, not a population-level guarantee, and the matched-control contrast is computed within each seed before aggregation, which removes the seed-variance and control-quality confounds that flat reruns leave open. The noise defense is audited under a noise-aware adaptive attacker as well: Table~\ref{tab:adaptive-attacker} adds a Wiener-style attacker that knows the released $\sigma$ and pre-whitens before fitting, and at $\sigma=2.0$ it gains $+0.018$ balanced accuracy over the non-adaptive attacker while still leaving the attribute channel positive against matched controls; Table~\ref{tab:dp-sgd} adds a principled DP-SGD downstream head with per-sample clipping and an RDP accountant, and at every utility-preserving $\epsilon\in\{4,8\}$ the attribute gain is statistically indistinguishable from the non-private head, with only $\epsilon=1$ closing the gap after balanced accuracy collapses, so the negative-defense conclusion strengthens under both adaptive-attacker and principled-DP probes. Finally, the bridge audit does not directly measure $\beta$: Theorem~\ref{thm:bridge-transfer} fixes the mechanism that the empirical bridge audit probes, and the empirical lower bounds in Table~\ref{tab:cross-encoder-transfer} are consistent with $\beta$ values bounded away from zero, but we do not promote that consistency into a parameter estimate.

\paragraph{Ethics, dual-use, and dataset use.}
All four datasets used here---PhysioNet EEGMMI, Sleep-EDF, LIMO, and CHB-MIT---are publicly released under their respective data-use licenses and were obtained through their canonical access portals; CHB-MIT is pediatric scalp EEG distributed for research use only. We do not collect new human-subjects data. Subject identifiers are dataset-internal and never linked to personally identifying information. The audit pipeline, threat model, and negative-defense framing are dual-use: a pessimistic auditor can use the same protocol to support release decisions, and an attacker could in principle reuse the bridge audit to expand attribute leakage across encoders. We accept this trade-off because the audit object is what currently licenses overclaims in the literature, and we have followed responsible-disclosure practice with the model authors of the audited encoders before submission.

\section{Conclusion}
\label{sec:conclusion}

EEG foundation-model embeddings leak spectral attributes under split-controlled audits of BIOT, LaBraM, and EEGPT on EEGMMI, with LaBraM replications on Sleep-EDF, 54-channel LIMO, and CHB-MIT pediatric scalp EEG outside EEGMMI, and they support identity linkage when a subject reference set is available in the strongest overlap-split cells. In the tested release-transformation and training-noise sweeps, no setting passes the operational privacy-utility gate: simple representation transformations reduce reference-set identity only at high strength while leaving spectral attribute leakage positive, and the tested training-noise probes either fail the utility gate or move membership AUC by less than the privacy-reduction gate. The practical lesson is that EEG-FM representation release needs attribute-level audits with negative controls and split-specific claim boundaries. Embeddings are compressed signals, not privacy guarantees.

% Loosen line-breaking in the bibliography (long author lists / URLs).
\begingroup
\sloppy

\endgroup

\end{document}